\documentclass[letterpaper]{article}[11pt]

\usepackage[a4paper,left=1.0in,right=1.0in,top=1.1in,bottom=1.3in]{geometry}
\usepackage{amsmath, amssymb, amsthm, framed}
\usepackage{fullpage}
\usepackage[colorlinks=true]{hyperref}
\usepackage{authblk}
\usepackage{tikz}
\usetikzlibrary{positioning}
\usetikzlibrary{shapes}
\usetikzlibrary{decorations.pathreplacing}
\newcommand{\BPP}{\mathsf{BPP}}
\newcommand{\PP}{\mathsf{PP}}
\newcommand{\UPP}{\mathsf{UPP}}
\newcommand{\THR}{\mathsf{THR}}
\newcommand{\MAJ}{\mathsf{MAJ}}
\newcommand{\GHR}{\mathsf{GHR}}

\newcommand{\xor}{\mathsf{XOR}}
\newcommand{\XOR}{\mathsf{XOR}}
\renewcommand{\mod}{\mathsf{MOD}}
\newcommand{\UTHR}{\mathsf{UTHR}}
\newcommand{\cq}{\mathsf{CQ}}

\newcommand{\R}{{\mathbb R}}
\newcommand{\E}{{\mathbb E}}
\renewcommand{\exp}{\mathsf{EXP}}
\newcommand{\PM}{\mathsf{PM}}
\newcommand{\mathify}[1]{\ifmmode{#1}\else\mbox{$#1$}\fi}
\newcommand{\bra}[1]{\{#1\}}
\newcommand{\abs}[1]{\mathify{\left| #1 \right|}}
\newcommand{\wh}{\widehat}
\newcommand{\disc}{\mathrm{disc}}
\newcommand{\corr}{\mathrm{corr}}
\newtheorem{theorem}{Theorem}[section]

\newtheorem{lemma}[theorem]{Lemma}

\newtheorem{definition}[theorem]{Definition}
\newtheorem{claim}[theorem]{Claim}
\newtheorem{fact}[theorem]{Fact}

\newtheorem{remark}[theorem]{Remark}
\newtheorem*{theorem*}{Theorem}

\newtheorem{cor}{Corollary}
\newcommand\restr[2]{{
  \left.\kern-\nulldelimiterspace 
  #1 
  \vphantom{\big|} 
  \right|_{#2} 
  }}

\begin{document}
\title{Dual polynomials and communication complexity of \textsf{XOR} functions}
\author{Arkadev Chattopadhyay\thanks{Partially supported by a Ramanujan fellowship of the DST. arkadev.c@tifr.res.in}}
\author{Nikhil S.~Mande\thanks{Supported by a DAE fellowship. nikhil.mande@tifr.res.in}}
\affil{School of Technology and Computer Science, TIFR, Mumbai}
\date{}                     
\setcounter{Maxaffil}{0}
\renewcommand\Affilfont{\itshape\small}
\maketitle

\begin{abstract}
We show a new duality  between the \emph{polynomial margin} complexity of $f$ and the \emph{discrepancy} of the function $f \circ \xor$, called an $\xor$ function. Using this duality, 
we develop polynomial based techniques for understanding the bounded error ($\BPP$) and the weakly-unbounded error ($\PP$) communication complexities of $\xor$ functions.  This enables us to show the following.

\begin{itemize}
\item A weak form of an interesting conjecture of Zhang and Shi\footnote{The full conjecture has just been reported to be independently settled by Hatami and Qian \cite{HQ17}.
However, their techniques are quite different and are not known to yield many of the results we obtain here.} \cite{SZ09} asserts that for symmetric functions $f : \bra{0, 1}^n \rightarrow \bra{-1, 1}$, the weakly unbounded-error complexity of $f \circ \xor$ is essentially characterized by
the number of points $i$ in the set $\{0,1, \dots,n-2\}$ for which $D_f(i) \ne D_f(i+2)$, where $D_f$ is the predicate corresponding to $f$. The number of such points is called the \emph{odd-even} degree of $f$.
We observe that a much earlier work of Zhang \cite{Zhang91} implies that the $\PP$ complexity of $f \circ \xor$ is $O(k \log n)$, where $k$ is the odd-even degree of $f$.
We show that the $\PP$ complexity of $f \circ \xor$ is $\Omega(k/ \log(n/k))$.

\item We resolve a conjecture of Zhang \cite{Zhang91} characterizing the Threshold of Parity circuit size of symmetric functions in terms of their odd-even degree.

\item We obtain a new proof of the exponential separation between $\PP^{cc}$ and $\UPP^{cc}$ via an $\xor$  function.

\item We provide a characterization of the \emph{approximate spectral norm} of symmetric functions, affirming a conjecture of Ada et al.~\cite{AFH12} which has several consequences (cf.~\cite{AFH12}).
This also provides a new proof of the characterization of the bounded error complexity of symmetric $\xor$ functions due to \cite{SZ09}.

\end{itemize}

Additionally, we prove strong $\UPP$ lower bounds for $f \circ \xor$, when $f$ is symmetric and periodic with period $O(n^{1/2-\epsilon})$, for any constant $\epsilon > 0$.
More precisely, we show that every such $\xor$ function has unbounded error complexity $n^{\Omega(1)}$, unless $f$ is constant or parity or its complement, in which case the complexity is just $O(1)$.
As a direct consequence of this, we derive new exponential lower bounds on the size of depth-2 threshold circuits computing such $\xor$ functions.
Our $\UPP$ lower bounds do not involve the use of linear programming duality.

\end{abstract}

\thispagestyle{empty}
\newpage
\setcounter{page}{1}
\section{Introduction}\label{sec:intro}

We consider three well known models of randomized communication, in all of which Alice and Bob use only \emph{private random coins}. Alice and Bob receive a pair of inputs $X \in \mathcal{X}$ and $Y \in \mathcal{Y}$ respectively. 
They want to jointly evaluate a function $f : \mathcal{X} \times \mathcal{Y} \to \{0,1\}$ on the pair $(X,Y)$ by using a communication protocol that minimizes the total \emph{cost} in the worst case. 
The protocol is probabilistic with the requirement that $\Pr\left[\Pi(X,Y) = f(X,Y)\right] \ge 1/2 + \epsilon$, where $\epsilon > 0$. Each of the three models specifies a basic requirement on $\epsilon$, 
and the goal of the players is to design an \emph{efficient} protocol meeting this requirement that minimizes the cost. Further, each model has its own cost function. 
Protocols are efficient if their cost is poly-logarithmic in $n$, the length of the inputs to Alice and Bob.  

Yao \cite{Yao79} introduced the model where the advantage $\epsilon$ needs to be a positive constant independent of the length of the inputs and the cost is the total number of bits communicated. 
The cost of the best protocol for computing a function in this model is called its \emph{bounded error} complexity. Paturi and Simon \cite{PS86} relaxed the requirement on advantage completely: $\epsilon$ only needs to be positive, 
but it can be otherwise decreasing arbitrarily with $n$. The complexity of $f$ in this model is called its \emph{unbounded error} complexity. 
Babai et al.~\cite{BFS86} introduced a semi-relaxed model whose power is sandwiched between the two above models: while the correctness requirement is the same as that in the unbounded error case, low advantage is penalised by
introducing a term in the cost function: the cost of a protocol is the sum of the total number of bits communicated \emph{and} $\log\big(1/\epsilon\big)$. The complexity of a function in this semi-relaxed model is called its
\emph{weakly-unbounded error} complexity. The set of functions that have efficient bounded, weakly-unbounded and unbounded error protocols are called $\BPP^{cc}, \PP^{cc}$ and $\UPP^{cc}$ respectively, closely borrowing terminology from 
standard Turing machine complexity classes. 

Clearly, $\BPP^{cc} \subseteq \PP^{cc} \subseteq \UPP^{cc}$. Set-Disjointness, denoted by DISJ, separates $\BPP^{cc}$ from $\PP^{cc}$ due to \cite{BFS86} and the following simple $\PP$ protocol of logarithmic cost:
Alice randomly chooses an index $i$ in $[n]$ and sends the value of $i$ and her $i$th bit to Bob. If both Alice and Bob have 1 as their $i$th bit, Bob outputs that they are not disjoint. Otherwise Alice and Bob output a random answer.
Thus, the weakly-unbounded error complexity of DISJ, commonly considered to be a hard function \cite{BFS86, Razborov92, KS92}, is exponentially smaller than its bounded error complexity.

There are fewer known strong lower bounds for the $\PP$ model than the bounded error one. This is partly explained by the fact that while techniques based on corruption and information theory yield lower bounds for bounded error model,
the $\PP$ model is exactly characterized by the stronger measure of discrepancy \cite{Klauck07}. Still, there are several functions for which discrepancy can be bounded. $\PP^{cc}$ was separated from $\UPP^{cc}$ in independent works of
Sherstov \cite{She08} and Buhrman et al.~\cite{BVW07}. Proving lower bounds for the unbounded error model, on the other hand, is even more difficult. The only known way for proving bounds here is lower bounding the
\emph{sign-rank} of the communication matrix \cite{PS86}. The sign rank of a real matrix $M$ with non-zero entries is the smallest number $r$ such that there exists a matrix $M'$ of rank $r$ and the same dimension as that of
$M$ such that each of its entries has the same sign as the corresponding one in $M$. Clearly, there is a matrix rigidity-like flavor to this definition, perhaps explaining the difficulty of estimating this quantity well.
In a beautiful and breakthrough work, Forster \cite{Forster01} managed to show that the Inner-Product (IP) function has high sign rank and consequently high unbounded-error complexity. The technique of Forster,
relating the spectral norm of a matrix to its sign rank, forms the basis for the few subsequent works on lower bounds for explicit functions in the model. In particular, Razborov and Sherstov \cite{RS10} and
Sherstov \cite{Sherstov11a} prove lower bounds on different functions making very interesting use of additional tools from approximation theory.

We consider a different class of functions in the unbounded error model. To explain this, let us introduce function composition. Given $f:\{0,1\}^n \to \{0,1\}$ and $g:\{0,1\}^{2b} \to \{0,1\}$,
we denote by $f \circ g$ the following composed function: its input is naturally viewed as a $2 \times bn$ matrix consisting of $n$ blocks each of which is a $2 \times b$ matrix.
Alice and Bob get the first and second row respectively of this matrix. We define $\big(f \circ g\big)\big(w_1,\ldots,w_n\big) = f\big(g(w_1),\ldots,g(w_n)\big)$, where $w_i$ is the $i$th block. Here $b$ is called the block length.
Thus, $\text{DISJ}$ is $\textsf{NOR} \circ \textsf{AND}$, with block length 1. Similarly, IP is the block size 1 function $\xor \circ \textsf{AND}$. Both have \textsf{AND} as the inner function but, as pointed out earlier,
have widely different unbounded error complexity. What makes \textsf{AND} functions, that is functions of the form $f \circ \textsf{AND}$ with block length 1, difficult? 

An important step towards understanding this was taken in the works of Sherstov \cite{She09,She11} and Shi and Zhu \cite{ShiZ09}. 
These papers reduced the task of proving lower bounds on the cost of both (quantum) bounded error and weakly unbounded error protocols for functions of the form $f \circ \textsf{AND}$ to that of analyzing
the approximability of $f$ by low degree real polynomials. This passage was achieved by making very elegant use of linear programming duality. This method spawned further progress in at least two directions.
One was the adaptation of the technique to multi-party communication complexity in \cite{Cha07,CA08,LS09a,Cha09}, resulting in the first super-polynomial lower bounds for Disjointness in the hard NOF model.
Using even more powerful approximation theoretic tools for polynomials, Sherstov \cite{She14} significantly improved these bounds. In another direction,
Razborov and Sherstov \cite{RS10} and Sherstov \cite{Sherstov11a} further demonstrated the power of these dual polynomial based techniques by analyzing the unbounded error complexity of $f \circ \textsf{AND}$
when $f$ is a certain $\text{AC}^0$ function or it is symmetric. In short, dual polynomial techniques provide a systematic way of analyzing the communication complexity of \textsf{AND} functions.
Besides these impressive developments, this approach relates to research on approximation theoretic questions on boolean functions, that are of independent interest (see for example \cite{BT17, Thaler16}).

There are essentially two inner functions of block length 1, \textsf{AND} and $\xor$. A natural example of an \textsf{XOR} function is $\mathsf{AND} \circ \xor$, better known as Equality.
However, even its bounded error (private coin) complexity is only $O(\log n)$, while its unbounded error complexity is just $O(1)$. In fact, in some contexts as discussed later in this work,
proving even $\PP$ lower bounds for $\xor$ functions seems more challenging than proving lower bounds for $\textsf{AND}$ functions. Interestingly, Sherstov \cite{She08} used an $\XOR$ function introduced by
Goldmann, H{\aa}stad and Razborov \cite{GHR92}, to separate $\PP^{cc}$ from $\UPP^{cc}$.  Zhang and Shi \cite{SZ09} characterized the
bounded error and quantum complexity of all symmetric \textsf{XOR} functions. Recently, Hatami, Hosseini and Lovett \cite{HHL16} nearly characterized the deterministic complexity of all \textsf{XOR} functions.
Even more recently, after an initial version of this manuscript containing weaker results was submitted, Hatami and Qian \cite{HQ17} have just reported settling a conjecture of Zhang and Shi \cite{SZ09}
on the unbounded error complexity of symmetric $\XOR$ functions. Both \cite{SZ09,HQ17} analyze $\XOR$ functions by finding simple reductions to appropriate \textsf{AND} functions. While such arguments are short, as commented by Ada et al.~\cite{AFH12},
it seems they do not provide new insights and techniques that can be applied more broadly to \textsf{XOR} functions.

In this work, we develop a dual polynomial based technique for analyzing $\XOR$ functions. 

Along the way, we discover an independently interesting general connection between the discrepancy of functions of the form $f \circ \XOR$ and the polynomial margin complexity of $f$.
Using this and other tools, we characterize the $\PP$ complexity of symmetric $\XOR$ functions and provide a new proof of the exponential separation between $\PP^{cc}$ and $\UPP^{cc}$ via an $\XOR$ function.
We further provide a new proof of the characterization of Zhang and Shi \cite{SZ09} of the bounded error complexity of symmetric $\XOR$ functions.
Our argument, unlike theirs, is based on a connection between the approximate spectral norm of $f$ and the bounded error communication complexity of $f \circ \xor$.
While this connection seems to have been first reported in the survey by Lee and Shraibman \cite{LS09}, as far we know, and as expressed in Ada et al.~\cite{AFH12}, 
it has not been used before this work in deriving explicit lower bounds on communication complexity. 

In the course of proving lower bounds on communication complexity, we obtain new results on two complexity measures of symmetric functions that are of independent interest. 
First, we characterize symmetric functions computable by quasi-polynomial size depth 2 boolean circuits of the form Threshold of Parity, resolving an old conjecture of Zhang \cite{Zhang91}.
Further, we characterize the approximate spectral norm of symmetric functions, confirming a conjecture of Ada et al.~\cite{AFH12}, which has various consequences (cf.~\cite{AFH12}).
We feel that these developments exhibit the potential of the dual polynomial based technique for proving lower bounds against $\XOR$ functions in general (that are not necessarily symmetric).

\subsection{Our Results}\label{sec:results}

In this section, we outline our main results.

\subsubsection{Polynomial complexity measures of symmetric functions}

In this section, we outline results we obtain by amplifying hardness of functions using the method of lifting as defined in Krause and Pudl\'{a}k.
We list applications of this `hardness amplification' to symmetric functions, which resolve conjectures by Ada et al.~\cite{AFH12} and Zhang \cite{Zhang91}.

For any function $f : \bra{-1, 1}^n \rightarrow \bra{-1, 1}$, let $f = \sum_{S \subseteq [n]}c_S \prod\limits_{i \in S}x_i$ be the unique multilinear expansion of $f$.
Define the \emph{weight} of $f$, denoted by $wt(f)$ to be $\sum_{S \subseteq [n]}\abs{c_S}$.
\footnote{Note that this notion coincides with $||\hat{f}||_1$, the spectral norm of $f$.  However, for the purposes of this paper, we shall use the former notation.}

\begin{definition}[Approximate weight]\label{defn: appwt}
Define the \textnormal{$\epsilon$-approximate weight} of a function $f : \bra{-1, 1}^n \rightarrow \bra{-1, 1}$, denoted by $wt_{\epsilon}(f)$ to be the weight of a minimum weight polynomial such that
for all $x \in \bra{-1, 1}^n, ~ \abs{p(x) - f(x)} < \epsilon$.
\footnote{This notion coincides with the notion of the \emph{$\epsilon$-approximate spectral norm} of $f$, denoted by $||\hat{f}||_{1, \epsilon}$, as defined in \cite{AFH12}.}
\end{definition}

\begin{definition}\label{defn: rf}
Let $F : \bra{-1, 1}^n \rightarrow \bra{-1, 1}$ be a symmetric function.  Define $r_0 = r_0(F), r_1 = r_1(F)$ to be the minimum integers $r_0'$ and $r_1'$ respectively, 
such that $r_0', r_1' \leq n/2$ and $D_F(i) = D_F(i + 2)$ for all $i \in [r_0', n - r_1')$.  Define $r = r(F) = \max\bra{r_0, r_1}$.
\end{definition}
\begin{definition}[Margin]\label{defn: margin}
The margin of a function $f : \bra{-1, 1}^n \rightarrow \bra{-1, 1}$ is defined as follows.
\[
m(f) \triangleq \max_{p : wt(p) \leq 1}\left(\min_{x \in \bra{-1, 1}^n}{p(x)f(x)}\right)
\]
Here, the maximum is only taken over those polynomials $p$ which sign represent $f$ everywhere.
\end{definition}

We prove the following powerful theorem which gives us lower bound tools against approximate weight, signed monomial complexity, and polynomial margin of symmetric functions.

\begin{theorem}\label{thm: liftsym}
Let $F : \bra{-1, 1}^n \rightarrow \bra{-1, 1}^n$ be any symmetric function.
\begin{enumerate}
\item If $r(F) \geq 5$, then there exists a universal constant $c_1 > 0$ such that
\[
\log(wt_{1/3}(F)) \geq c_1 \cdot r(F).
\]
\item If $k = \deg_{oe}(F) \geq 16$, then there exists a universal constant $c_2$ such that
\[
\mathrm{mon}_{\pm}(F) \geq 2^{c_2 \cdot k / \log(n/k)}
\]
\item If $k = \deg_{oe}(F) \geq 16$, then there exists a universal constant $c_3$ such that
\[
m(F) \leq \frac{1}{2^{c_3 \cdot k / \log(n/k)}}
\]
\end{enumerate}
\end{theorem}

We also use Part 1 of Theorem \ref{thm: liftsym}, to prove the following theorem, posed as a conjecture by Ada et al.~\cite{AFH12}.

\begin{theorem}[Conjecture 1 in \cite{AFH12}]\label{thm: afhsolve}
For any symmetric function $F : \bra{-1, 1}^n \rightarrow \bra{-1, 1}$, there exist universal constants $c_0, c_1 > 0$ such that
\[
c_0 \cdot r(F)\log\left(\frac{n}{r(F)}\right) \geq \log wt(F) \geq \log wt_{1/3}(F) \geq c_1 \cdot r(F)
\]
\end{theorem}

This has several consequences (cf.~\cite{AFH12}), which we do not state here.
We also resolve the following conjecture by Zhang \cite{Zhang91}.
\begin{theorem}[Conjecture 1 in \cite{Zhang91}]\label{thm: zhangsolve}
A symmetric function $f : \bra{-1, 1}^n \rightarrow \bra{-1, 1}$ is computable by a quasi-polynomial size Threshold of Parity circuit if and only if its odd-even degree is $\log^{O(1)}n$.
\end{theorem}

\subsubsection{$\PP$ complexity}\label{sec:PP}
In this section, we list our results regarding the $\PP$ complexity of $\xor$ functions.

Our main tool for analyzing the discrepancy of $\XOR$ functions is a tight relationship (upto constant factors) between $\disc(f \circ \xor)$ and $m(f)$. We derive this using linear programming duality.

\begin{theorem}[Polynomial Margin-Discrepancy theorem]\label{thm: equiv}
Let $f \rightarrow \bra{-1, 1}^n \rightarrow \bra{-1, 1}$.
\[
m(f) \leq m(f \circ \xor) \leq 4\disc(f \circ \xor) \leq 4m(f)
\]
\end{theorem}
The proof of Theorem \ref{thm: equiv} shows that the discrepancy of every $\XOR$ function is attained on a lifted distribution. Indeed, our Margin-Discrepancy Theorem is a lifting theorem for XOR functions
that primarily reduces the task of lower bounding the discrepancy of $f \circ \xor$ with that of establishing bounds on the polynomial margin of $f$. The second task is likely easier using tools from approximation theory.
There is a compelling parallel here with the Degree-Discrepancy Theorem of Sherstov \cite{She09}. This theorem has yielded a methodical way of proving discrepancy bounds for $f \circ \PM$ by lower bounding the \emph{sign degree} of $f$,
where $\PM$ denotes the pattern matrix gadget, and is defined formally in Section \ref{sec: prelims}. This has led to much progress in understanding the communication complexity of $\textsf{AND}$ functions
(for example, \cite{Cha07, CA08, She11, Sherstov11a}). We believe our polynomial Margin-Discrepancy Theorem will yield a unified approach in making similar progress for XOR functions. As evidence of this, we provide two applications of this theorem.

The first shows that the $\PP$ complexity of functions of the form $F \circ \xor$ for symmetric $F$ is essentially the odd-even degree of $F$ (upto polylogarithmic factors) as predicted by the conjecture of Shi and Zhang.

\begin{theorem}\label{thm: sym_xor}
Let $F : \bra{-1, 1}^{4n} \rightarrow \bra{-1, 1}$ be any symmetric function, and let $r \geq 4$ be its odd-even degree.  Then, there exists a universal constant $c > 0$ such that $\PP(F \circ \xor) \geq cr/\log(n/r)$ where $\PP(F \circ \xor)$ denotes the $\PP$ complexity of $F \circ \xor$.
\end{theorem}

To prove the above, Theorem \ref{thm: equiv} sets the goal of establishing a bound on the margin complexity of symmetric functions with large odd-even degree.
We do this by showing that symmetric functions with large odd-even degree can be projected onto a certain lift of symmetric functions with high sign degree.
This enables us to work with the more convenient notion of sign degree rather than odd-even degree of symmetric functions. 

As another application of our Margin-Discrepancy connection, we provide a new proof of the separation of $\PP^{cc}$ from $\UPP^{cc}$. We do this by proving that an $\XOR$ function,
almost identical to the $\GHR$ function (cf.~\cite{GHR92}) has exponentially small discrepancy. It is well known that this function has very efficient $\UPP$ protocols. 
We define the $\GHR$ function formally in Section \ref{subsec: comm}.

\begin{theorem}\label{thm: main}
\hspace{2em}
\begin{enumerate}
 \item There exists a linear threshold function $f : \bra{-1, 1}^n \rightarrow \bra{-1, 1}$ and an absolute constant $c > 0$ such that $\PP(f \circ \xor) \geq cn$.
 \item $\PP(\GHR) \geq \Omega(\sqrt{n})$.
\end{enumerate}
\end{theorem}

\subsubsection{$\BPP$ complexity}

Using linear programming duality and the generalized discrepancy method (Theorem \ref{thm: gendisc}), we give a simple alternate proof of the following result from \cite{LS09}.
\begin{theorem}\label{thm: bpp}
For any function $f : \bra{-1, 1}^n \rightarrow \bra{-1, 1}$, there exists a universal constant $c > 0$ such that
\[
R_{1/3}(f \circ \xor) \geq c\log\left(wt_{1/3}(f) - 4\right).
\]
\end{theorem}
\begin{remark}
In fact, lower bounds on $wt_{1/3}(f)$ yield lower bounds on the bounded error quantum communication complexity of $f \circ \xor$).
\end{remark}
Although Theorem \ref{thm: bpp} was known from \cite{LS09}, to the best of our knowledge, ours is the first work to use this technique to prove lower bounds for explicit functions.

Using Part 1 of Theorem \ref{thm: liftsym} in conjunction with Theorem \ref{thm: bpp} provides an alternate proof of the following result of Zhang and Shi \cite{SZ09}.
\begin{theorem}[\cite{SZ09}]\label{thm: SZ}
Let $F : \bra{-1, 1}^n \rightarrow \bra{-1, 1}$ be any symmetric function.  Then, $R_{1/3}(F \circ \xor) = \Omega(r(F))$.
\end{theorem}

Blais et al.~\cite{BBG14} also provided an alternate proof to Theorem \ref{thm: SZ} by showing a lower bound on the information complexity of symmetric $\xor$ functions (this however, does not imply quantum lower bounds).

\subsubsection{$\UPP$ complexity}   \label{sec:UPP}
We consider the $\UPP$ complexity $f \circ \xor$ when $f$ is symmetric and periodic. More precisely,

\begin{definition}[\textsf{MOD} functions and simple accepting sets]\label{defn: mod}
A function $f: \bra{0, 1}^n \rightarrow \bra{-1, 1}$ is called a \textsf{MOD} function if there exists a positive integer $m < n$ and an `accepting' set $A \subseteq [m]$ such that
\begin{align*}
f(x) = \begin{cases}
        -1 & \sum\limits_{i = 1}^nx_i \equiv k \text{ mod } m \text{ for some } k \in A\\
        1 & \text{otherwise}
       \end{cases}
\end{align*}
We write $f = \mod_m^A$.
We call an accepting set $A$ \emph{simple} if $\mod_m^A$ either represents the constant 0 function, constant 1 function, or the parity function or its negation.
We also call the corresponding predicate \emph{simple} in this case.
\end{definition}

We now state our main result regarding unbounded error communication below:

\begin{theorem}\label{thm: actual_main}
For any integer $m \geq 3$, express $m = j2^k$ uniquely, where $j$ is either odd or 4, and $k$ is a positive integer.
Then for any non-simple $A$,
\[
\UPP(\mod_{m}^A \circ \xor) \geq \Omega\left(\frac{n - km}{jm}\right) - \frac{2j\log j}{m}
\]
where $\UPP(f)$ denotes the unbounded error communication complexity of $f$.
\end{theorem}

\begin{remark}
A very recent result of Hatami and Qian \cite{HQ17} subsumes Theorem \ref{thm: actual_main}.  However, their result is based on a simple reduction to symmetric \textsf{AND} functions, whose unbounded error complexity has been tightly
characterized by Sherstov \cite{Sherstov11a} using sophisticated tools from approximation theory.  Our result, on the other hand, is based on first principles using Fourier analysis of boolean functions.
\end{remark}

The above implies that the \textsf{XOR} function corresponding to a symmetric and periodic $f$ with period $O(n^{1/2 - \epsilon})$, for some constant $\epsilon > 0$, has unbounded-error complexity
$n^{\Omega(1)}$ as long as $f$ is neither constant nor Parity nor its complement.

A well known consequence of proving unbounded error lower bounds against $f$ is a lower bound for the size of depth-2 circuits of the form $\THR \circ \mathsf{L_{comm}}$ computing $f$ where $\mathsf{L_{comm}}$
denotes the class of functions with low deterministic communication complexity.
As a result, Theorem \ref{thm: actual_main} implies that in particular, $\mod_m^A \circ \xor$ is not in polynomial sized $\THR \circ \mathsf{SYM}$ circuits, which we formally state in Theorem \ref{thm: ckt_lb}.
This generalizes a result of Zhang \cite{Zhang91} and Krause and Pudl\'{a}k \cite{KP97} who showed, among other things, that $\mod_p^{\bra{0}}$ cannot be computed by polynomial sized Threshold of Parity circuits.

\begin{figure}
\begin{center}
\tikzstyle{vis}=[draw=black,ellipse,]
\begin{tikzpicture}[transform shape, scale=0.9]
\node[vis] (fop) at (3,0) {$f^{op}$};
\node[vis] (f) at (-2,0) {$f$}
edge[->] node[below, midway,align=center,text width=3cm] {\small degree-monomial amplification}(fop);
\node[vis] (F) at (3,2) {$F$}
edge[->] node[midway, right] {\small projection} node[midway,left,align=right] {Lemma~\ref{lem: thr_lift}\\Lemma~\ref{lemma: symm_lift}} (fop);
\node[vis] (Fxor) at (9,2) {$F \circ \mathsf{XOR}$}
edge[<-] node[midway,above] {\small lift} (F);
\node at (9,2.7) {communication game};

\node (BPP) at (9,0.5) {$\BPP$ complexity};
\node (PP) at (9,-1) {$\PP$ complexity};

\node[align=right,anchor=east] (apdegree) at (-1,-1.5) {approximate degree};
\node[align=right,anchor=east,text width=5.8cm] (lerror) at (-1,-3) {large-error high-degree approximation};
\node[align=right,anchor=east] (signdeg) at (-1,-4.5) {sign degree};

\node[align=left,anchor=west] at (2,-1.5) {approximate weight}
edge[<-] (apdegree)
edge[->,dashed] node[midway,above,sloped] {Theorem~\ref{thm: bpp}} (BPP);
\node[align=left,anchor=west] at (2,-3) {polynomial margin}
edge[<-] (lerror)
edge[->,dashed] node[midway,above,sloped] {Theorem~\ref{thm: equiv}} (PP);

\node[align=left,anchor=west] at (2,-4.5) {signed monomial complexity}
edge[<-] node[midway, above] {\cite{KP97}} (signdeg);

\draw [decorate,decoration={brace,amplitude=5pt,mirror}] (-4.6,-1.3) -- node[midway, left=0.3cm] {Lemma~\ref{lem: lift}} (-4.6,-4.65);
\end{tikzpicture}
\caption{General framework}
\label{fig: fig}
\end{center}
 
\end{figure}
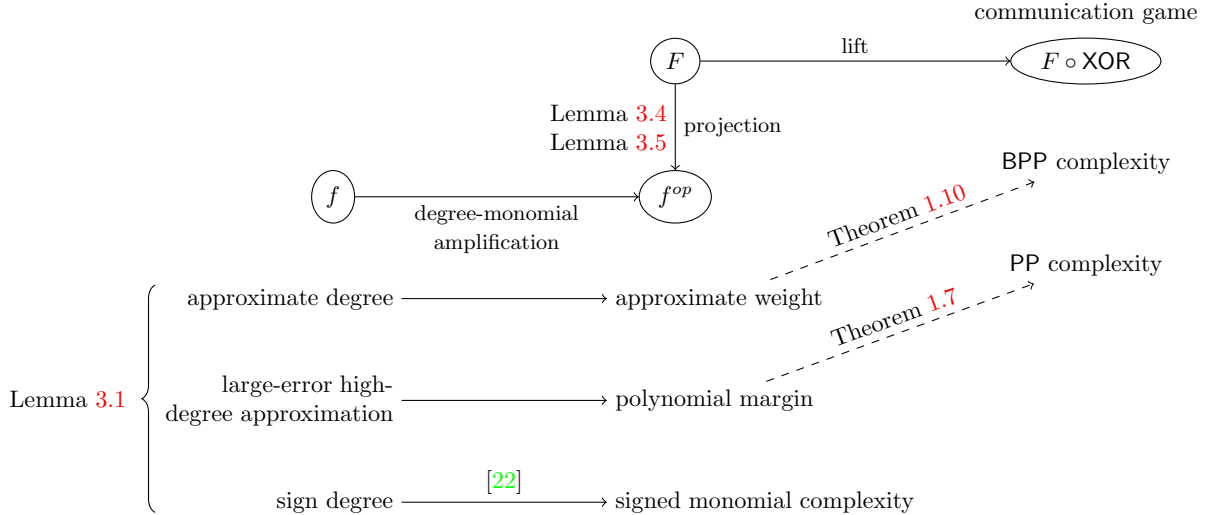

\subsection{Proof outline}

Our proof strategy is depicted in Figure \ref{fig: fig}. First, we use a nice idea due to Krause and Pudl{\'a}k \cite{KP97}, who showed that if a function $f$ has high sign degree, then a certain \emph{lift} of that function, 
denoted by $f^{op}$ has high sign monomial complexity. We observe that their argument can be easily adapted to show a more general result. In particular, our Lemma \ref{lem: lift} shows that the hardness of $f$ for \emph{low degree} polynomials,
with respect to natural notions like uniform approximation and sign representation, gets amplified to corresponding hardness of $f^{op}$ for \emph{sparse (low weight)} polynomials. 
Next, we observe that LP duality implies, via Theorems \ref{thm: equiv} and \ref{thm: bpp}, that such hardness of a function $F$ against sparse polynomials translates to the hardness of $F \circ \xor$ 
for the appropriate randomized ($\BPP, \PP$) communication model. The main problem at this point is to understand how $F$ relates to an appropriately hard $f^{op}$. In particular, our interest is when $F$ is a symmetric function or a linear halfspace.
These functions do not seem to have the structure of a lifted function $f^{op}$. 

At this point, inspired by the work of Krause \cite{Krause06}, we make a simple but somewhat counter-intuitve observation that turns out to be crucial. A function $g$ is called a monomial projection of $h$, 
if $g$ can be obtained by substituting each input variable of $h$ with a monomial in variables of $g$. What is nice about such projections is that for the polynomial sparsity measures (Lemma~\ref{lem: monproj}) that are relevant for us,
the complexity of $g$ is upper bounded by that of $h$.  We observe (Lemma~\ref{lemma: symm_lift} and Lemma~\ref{lem: thr_lift}) that if $f$ is a symmetric (linear threshold) function, then there exists a symmetric (linear threshold) function $F$
such that $f^{op}$ is a monomial projection of $F$. Moreover, the combinatorial parameters of $f$ that caused its hardness against low-degree polynomials, nicely translate to combinatorial parameters of $F$
that have been conjectured to cause hardness of $F$ against sparse (low weight) polynomials. By our LP duality theorems, these result in the hardness of $F \circ \xor$ against randomized communication protocols as well. 

The above describes the general framework of our passage from polynomials to communication protocols. We describe below the particular instantiations of this framework for each of the lower bounds that we prove.

\subsubsection{$\PP$ complexity}

We prove two main results regarding $\PP$ complexity by upper bounding margin complexity. The first is to reprove an exponential separation of $\PP$ protocols from those of $\UPP$, making use of the above framework. 
For this, it is natural to prove a strong $\PP$ lower bound against a function of the type $F \circ \xor$ where $F$ is a linear threshold function.
Proving a polylogarithmic $\UPP$ upper bound for such a function is straightforward. However, precisely this feature of $F$ makes it difficult to prove a strong $\PP$ lower bound.
Goldmann et al.~\cite{GHR92} used an ingenious specialized argument directly establishing that the discrepancy is small. We, on the other hand, use Theorem \ref{thm: equiv} which directs us in proving that $F$ must have small margin complexity.
The challenge here is to prove a strong unrestricted degree margin complexity lower bound against a function with sign degree just 1.
We use a variety of techniques to prove this. First, we use a result of Sherstov, Theorem \ref{thm: sherstov}, which states that there exists a linear threshold function $f$ which requires linear degree to approximate uniformly,
even with error inverse exponentially close to 1. Second, we use lifting as depicted in Figure \ref{fig: fig} to show that $f^{op}$ has a small upper bound on the (unrestricted degree) margin complexity.
We then use our monomial projection lemma for threshold functions, Lemma~\ref{lem: thr_lift}, to embed such a lifted function in a linear threshold function $F$ without blowing up the weights too much.  
Finally, we exploit the fact that the Universal Threshold function ($\UTHR$) embeds any other threshold function with at most a quadratic loss in number of variables. 
The last step of considering the $\UTHR$ is needed only to provide an \emph{explicit} exponential separation of $\PP$ and $\UPP$.

As a second application of our framework to $\PP$ complexity, we prove Theorem \ref{thm: sym_xor}, which states that the $\PP$ complexity of $F \circ \xor$ is essentially the odd-even degree of $F$ when it is symmetric.
The main challenge here is to work with the notion of odd-even degree, which has no immediate algebraic interpretation as opposed to sign degree. 
Lemma~\ref{lemma: symm_lift} solves this by essentially showing that there exists a symmetric $f$ whose sign degree corresponds to the odd-even degree of $F$, such that $f^{op}$ is a monomial projection of $F$.
Finally, our polynomial hardness amplification Lemma~\ref{lem: lift} shows that the margin of $f^{op}$ must be small if the base function $f$ has large sign degree.

\subsubsection{$\BPP$ complexity and approximate weight}
We first make a simple observation that the polynomial margin of a function $F$ equals its \emph{threshold weight}, as defined in Definition \ref{defn: thrwt}.
Just as the notion of threshold degree inspires the natural notion of approximate degree, threshold weight inspires the definition of approximate weight as in Definition \ref{defn: appwt}.
In Section \ref{sec: bdd}, we consider a linear program capturing the (1/3)-approximate weight of a symmetric function $F : \bra{-1, 1}^n \rightarrow \bra{-1, 1}$.
Using linear programming duality and the generalized discrepancy method, we show in Theorem \ref{thm: bpp} that $\log wt_{1/3}(F)$ is a lower bound on the bounded error communication complexity of $F \circ \xor$.

The general framework of Figure \ref{fig: fig} then prescribes us to find a suitable symmetric $f$ such that $f^{op}$ has large approximate weight and is a monomial projection of $F$.
Lemma~\ref{lemma: symm_lift} provides such a monomial projection in which the combinatorial quantity $r(F)$ corresponds to another combinatorial quantity $\Gamma(f)$, which is defined in Section~\ref{sec: prelims}.
Paturi's Theorem \cite{Paturi92} shows that $\Gamma(f)$ characterizes the approximate degree of $f$. The polynomial hardness amplification of Figure \ref{fig: fig}, via Lemma~\ref{lem: lift}, implies that $f^{op}$, and therefore $F$, 
has large approximate weight. This already proves Theorem~\ref{thm: afhsolve} which was conjectured by Ada et al.~\cite{AFH12}. Moreover, Theorem \ref{thm: bpp} implies the hardness of $F \circ \xor$ against bounded error protocols.

\subsubsection{$\UPP$ complexity}
We remark here that, although a very recent independent result of Hatami and Qian \cite{HQ17} subsumes our results on $\UPP$ complexity of symmetric $\xor$ functions,
our methods vary vastly from theirs.  We prove our lower bounds from first principles, and do not make a reduction to Sherstov's result \cite{Sherstov11a} on symmetric \textsf{AND} functions.
Interestingly, our $\UPP$ lower bounds are not obtained via linear programming duality, as opposed to our $\PP$ and $\BPP$ lower bounds.

The starting point of our work in proving $\UPP$ lower bounds is a modification of Forster's theorem \cite{Forster01} by Forster et al. \cite{FKLMSS01} who relate the sign-rank of a function
$f: \bra{0, 1}^n \times \bra{0, 1}^n \rightarrow \mathbb{R}$ in terms of the minimum value taken by $f$ and the spectral norm of the communication matrix of $f$.
Informally, the unbounded error complexity of $f$ is large if the minimum value taken by it is not too small, and the spectral norm is small.
Refer to Theorem \ref{thm: forster} for details.  We then note in Lemma \ref{lem: xor_eigenvalues} that the spectral norm of $f \circ \xor$ is just a scaling of the maximum Fourier coefficient of $f$.
Observe that $\mod_3^{\bra{0}}$ has a large principal Fourier coefficient even though the other coefficients are inverse exponentially small.  Thus, one cannot use Theorem \ref{thm: forster} directly.
Next, we prove in Theorem \ref{thm: sufficient} that if the $L_1$ mass of a subset of the Fourier coefficients of $f$ is sufficiently bounded away from 1, and the remaining coefficients are sufficiently small, 
we can still obtain a strong unbounded error lower bound for $f \circ \xor$. We then analyze the Fourier coefficients of $\mod$ functions, to show that they satisfy the above properties, 
and this helps us prove lower bounds for $\mod_m^A \circ \xor$ for odd integers $m$ with values upto $O(n^{1/2 - \epsilon})$ as long as $\mod_m^A$ does not represent a constant or parity function.
This still does not prove hardness for all $\mod$ functions with period at most $O(n^{1/2 - \epsilon})$ since it can be proved, for example, 
$\abs{\wh{\mod_4^{\bra{0}}}(\emptyset)} + \abs{\wh{\mod_4^{\bra{0}}}([n])} = 1$, thus not allowing us to use Theorem \ref{thm: sufficient}.
To handle this case, we make two crucial observations.  One is that setting a few variables (which we can view as shifting the accepting set by a small amount) does not change the 
unbounded error communication complexity of $\mod_m^A \circ \xor$ by much. The second is the fact that the unbounded error complexity of $f \oplus g$ is at most the unbounded error complexity of $f$ plus that of $g$.
Armed with these facts, we use a shifting and XORing trick that enables us to reduce the modulus of the target $\mod_m^A$ function to either 4 or a prime without using too large or too many shifts, or too many XORs.
We then use induction on $m$ to finish the proof of our main theorem regarding unbounded error communication (Theorem \ref{thm: actual_main}).

\section{Preliminaries}\label{sec: prelims}

We provide the necessary preliminaries in this section.

Note that in the following definitions, we interchangeably use the view of the input variables being $\bra{-1, 1}$ valued, and $\bra{0, 1}$ valued.
For most of our results regarding the discrepancy of $\xor$ functions, we view the input variables as $\bra{-1, 1}$ valued, whereas we view the inputs as $\bra{0, 1}$ valued
while dealing with the unbounded error model.  In general, $0$ corresponds to $1$, and $1$ corresponds to $-1$ in the two views.
We denote the \emph{Hamming weight} of a string $x \in \bra{0, 1}^n$ ($\bra{-1, 1}^n$) to be the number of variables set to $1$ ($-1$) in $x$.

\subsection{Types of functions}

A function $f : \bra{-1, 1}^n \rightarrow \bra{-1, 1}$ is called symmetric if $f(x_1, \dots, x_n) = f(x_{\sigma(1)}, \dots, x_{\sigma(n)})$ for all $\sigma \in S_n$ where $S_n$ denotes the set of all permutations on $n$ elements.
The value taken by a symmetric function on an input only depends on the Hamming weight of the input.
For a symmetric function $f: \bra{-1, 1}^n \rightarrow \bra{-1, 1}$, define its \textit{spectrum} or \textit{predicate} $D_f: \bra{0, 1, \dots, n} \rightarrow \bra{-1, 1}$ by $D_f(i) = f(x)$
where $x \in \bra{-1, 1}^n$ is such that there are $i$ many variables in $x$ taking the value $-1$.
Note that the spectrum (predicate) of a symmetric function is well defined.
Define the \emph{odd-even} degree of a symmetric function $f$, which we denote by $\deg_{oe}(f)$, to be $\abs{i \in \bra{0, 1, \dots, n - 2} : D_f(i) \neq D_f(i + 2)}$.

\begin{definition}[\textsf{XOR} functions]
A function $F: \bra{0, 1}^n \times \bra{0, 1}^n \rightarrow \bra{-1, 1}$ is said to be an $\xor$ function if there exists a function 
$f: \bra{0, 1}^n \rightarrow \bra{0, 1}$ such that $F(x_1, \dots, x_n, y_1, \dots, y_n) = f(x_1 \oplus y_1, \dots, x_n \oplus y_n)$  for all $x_1, \dots, x_n, y_1, \dots y_n \in \bra{0, 1}$.  We use the notation $F = f \circ \xor$.
\end{definition}

\begin{definition}[Threshold functions]
Define a function $f: \bra{0, 1}^n \rightarrow \bra{-1, 1}$ to be a linear threshold function if there exist integer weights $a_1, \dots, a_n$ such that for all inputs $x = (x_1, \dots, x_n) \in \bra{0, 1}^n$,
$f(x) = sgn\left(\sum_{i = 1}^na_ix_i\right)$.  Let $\THR$ denote the class of all such functions.
Let $\MAJ$ denote the class of linear threshold functions whose weights are polynomially bounded in $n$.
\end{definition}

\begin{definition}[Universal threshold]\label{defn: univ}
Define a class of threshold functions, $U_{l, k} : \bra{\bra{-1, 1}^{k}}^l \rightarrow \bra{0, 1}$ defined by
\[
U_{l, k}(x_{1, 1}, \dots, x_{1, k}, \dots, x_{l, 1}, \dots, x_{l, k}) = sgn\left(\sum\limits_{i = 1}^k \sum\limits_{j = 1}^l 2^ix_{i, j} + \frac{1}{2}\right)
\]
\end{definition}
The constant term $\frac{1}{2}$ is added to ensure that the sum inside the brackets is never 0.
\begin{fact}[Minsky and Papert \cite{MP69}]\label{fact: quad_blowup}
$U_{l, k}$ is universal in the sense that any linear threshold function on $n$ variables occurs as a subfunction of $U_{l, k}$ for some $l, k \in O(n\log n)$.
\end{fact}
We use the notation $\UTHR$ to denote such a function.

\subsection{Fourier analysis}
Consider the vector space of functions from $\bra{0, 1}^n$ to $\mathbb{R}$, equipped with the following inner product.
\[
\langle f, g \rangle = \E_{x \in \bra{0, 1}^n}f(x)g(x) = \frac{1}{2^n}\sum\limits_{x \in \bra{0, 1}^n}f(x)g(x)
\]
Define characters $\chi_S$ for every $S \subseteq [n]$ by $\chi_S(x) = (-1)^{\sum_{i \in S}x_i}$.
The set $\bra{\chi_S : S \subseteq [n]}$ forms an orthonormal basis for this vector space.
Thus, every $f: \bra{0, 1}^n \rightarrow \mathbb{R}$ can be uniquely written as $f= \sum\limits_{S \subseteq [n]}\wh{f}(S)\chi_S$ where
\begin{equation}\label{eqn: fourier_coefficients}
\wh{f}(S) = \langle f, \chi_S\rangle = \E_{x \in \bra{0, 1}^n}f(x)\chi_S(x)
\end{equation}

\subsection{Polynomials}
For a polynomial of weight 1, say $p$, which sign represents a function $f$, we say that $p$ represents $f$ with a margin of value $\min_{x \in \bra{-1, 1}^n}f(x)p(x)$.
Let us also define a notion of the error in a pointwise approximation of a function by low degree polynomials.  This notion is studied widely in classical approximation theory, see \cite{ShiZ09, She11, Thaler16} for example.
Note that we do not restrict the weight of the approximating polynomial in this case.
\begin{equation}
\varepsilon_d(f) \triangleq \min_{p : deg(p) \leq d}\left(\max_{x \in \bra{-1, 1}^n}{\abs{p(x) - f(x)}}\right)
\end{equation}

Sherstov \cite{She16} proved that there exists a linear threshold function which cannot be approximated well, even by large degree polynomials.
\begin{theorem}[\cite{She16}, Cor 3.3]\label{thm: sherstov}
There exists a linear threshold function $f : \bra{-1, 1}^n \rightarrow \bra{-1, 1}$ and an absolute constant $c > 0$ such that 
\[
\varepsilon_{cn}(f) > 1 - 2^{-cn}
\]
Moreover, the weights of the coefficients in the function have magnitude at most $2^n$.
\end{theorem}

\begin{definition}[Approximate degree]\label{def: appdeg}
For any function $f : \bra{-1, 1}^n \rightarrow \bra{-1, 1}$ and polynomial $p : \bra{-1, 1}^n \rightarrow \R$, we say that $p$ approximates $f$ to error $\epsilon$ if
for all $x \in \bra{-1, 1}^n, ~ \abs{p(x) - f(x)} \leq \epsilon$.  The $\epsilon$-approximate degree of $f$, denoted $\widetilde{\deg}_{\epsilon}(f)$ is the minimum degree of a polynomial $p$ which approximates $f$ to error $\epsilon$.
\end{definition}

\begin{definition}[Signed monomial complexity]\label{defn: smc}
The \emph{signed monomial complexity} of a function $f : \bra{-1, 1}^n \rightarrow \bra{-1, 1}$, denoted by $\mathrm{mon}_{\pm}(f)$ is the minimum number of monomials required by a polynomial $p$ to sign represent $f$ on all inputs.
\end{definition}
Note that the signed monomial complexity of a function $f$ exactly corresponds to the minimum size Threshold of Parity circuit computing it.

\begin{theorem}[\cite{Zhang91}]\label{thm: zhang_ub}
Let $f : \bra{-1, 1}^n \rightarrow \bra{-1, 1}$ be a symmetric boolean function such that $\deg_{oe}(f) = \log^{O(1)}n$.  Then, $f$ can be computed by a quasi-polynomial size Threshold of Parity circuit.
\end{theorem}

The following is a result by Paturi \cite{Paturi92} which gives us tight bounds on the approximate degree of symmetric functions.
\begin{theorem}[\cite{Paturi92}]\label{thm: paturi}
For any symmetric function $f : \bra{0, 1}^n \rightarrow \bra{-1, 1}$, define the quantity $\Gamma(f) = \min\bra{\abs{2k - n + 1} : D_f(k) \neq D_f(k + 1) \text{ and } 0 \leq k \leq n - 1}$.
Then,
\[
\widetilde{\deg}_{2/3}(f) = \Theta(\sqrt{n(n - \Gamma(f))})
\]
\end{theorem}

\begin{definition}
For functions $f, g : \bra{-1, 1}^n \rightarrow \bra{-1, 1}$ and a distribution $\nu$ on $\bra{-1, 1}^n$, define the correlation between $f$ and $g$ under the distribution $\nu$ to be
\[
\corr_\nu(f, g) = \E_\nu[f(x)g(x)]
\]
\end{definition}

\begin{definition}[Threshold weight]\label{defn: thrwt}
Define the \emph{threshold weight} of a function $f : \bra{-1, 1}^n \rightarrow \bra{-1, 1}$, denoted by $wt_{\pm}(f)$ to be the weight of a minimum weight real polynomial $p$ such that $p(x)f(x) \geq 1$ for all $x \in \bra{-1, 1}^n$.
\end{definition}
Note that this definition differs from the notion of more widely studied notion of threshold weight (see for example \cite{Krause06}, \cite{She11}, \cite{BT15}), where the coefficients of $p$ are restricted to be integer valued.
It is convenient for us to work with the notion as defined in Definition \ref{defn: thrwt} because of its following relationship with the polynomial margin, which can be easily verified.

\begin{lemma}\label{lem: margin_wt_equiv}
For any function $f : \bra{-1, 1}^n \rightarrow \bra{-1, 1}$,
\[
m(f) = \frac{1}{wt_{\pm}(f)}
\]
\end{lemma}

The following theorem was proved by Ada et al.~\cite{AFH12}, which characterizes the weight of a symmetric function.
\begin{theorem}[\cite{AFH12}]\label{thm: afh}
For any symmetric function $f : \bra{-1, 1}^n \rightarrow \bra{-1, 1}$,
\[
\log(wt(f)) = \Theta\left(r(f)\log\left(\frac{n}{r(f)}\right)\right)
\]
\end{theorem}

\subsection{Communication complexity}\label{subsec: comm}

We now recall some notions from communication complexity.

In the models of communication of our interest, two players, say Alice and Bob, are given inputs $X \in \mathcal{X}$ and $Y \in \mathcal{Y}$ for some finite input sets $\mathcal{X}, \mathcal{Y}$, they have access to \emph{private} randomness and they wish to compute a given function
$f: \mathcal{X} \times \mathcal{Y} \rightarrow \bra{-1, 1}$.  Unless mentioned otherwise, we use $\mathcal{X} = \mathcal{Y} = \bra{0, 1}^n$.
Alice and Bob communicate according to a protocol which has been fixed in advance.  The cost of a protocol is the maximum number of bits communicated on the worst case input.
A probabilistic protocol $\Pi$ computes $f$ with advantage $\epsilon$ if the probability that $f$ and $\Pi$ agree is at least $1/2 + \epsilon$ for all inputs.
Denote the cost of the best such protocol to be $R_{\epsilon}(f)$.  Note that we deviate from the notation used in \cite{KN97}.
Define the following measures of complexity of $f$.
\[
\PP(f) = \min_\epsilon\left(R_\epsilon(f) + \log\left(\frac{1}{\epsilon}\right)\right)
\]
and
\[
\UPP(f) = \min_\epsilon(R_\epsilon(f)).
\]
The latter quantity was introduced by Paturi and Simon \cite{PS86}, and we call it the unbounded error communication complexity of $f$.
The former adds a penalty term depending on the advantage, and was proposed by Babai et al.~\cite{BFS86}.  We refer to this cost as the weakly-unbounded error communication complexity of $f$.
These measures give rise to the following communication complexity classes \cite{BFS86}.

\begin{definition}\label{defn: PPUPP}
\[
\PP^{cc}(f) \equiv \bra{f: \PP(f) = \textnormal{polylog}(n)}
\]
\[
\UPP^{cc}(f) \equiv \bra{f: \UPP(f) = \textnormal{polylog}(n)}
\]
\end{definition}

Define the discrepancy of a rectangle $S \times T$ under a distribution $\lambda$ on $\bra{-1, 1}^n \times \bra{-1, 1}^n$ as follows.
\begin{definition}[Discrepancy]\label{defn: disc}
\[
\disc_{\lambda}(S \times T, f) = \sum\limits_{(x, y) \in S \times T}f(x, y)\lambda(x, y)
\]
The discrepancy of $f$ under a distribution $\lambda$ is defined as
\[
\disc_{\lambda}(f) = \max_{S \subseteq [n], T \subseteq [n]} \disc_{\lambda}(S \times T, f)
\]
and the discrepancy of $f$ is defined to be
\[
\disc(f) = \min\limits_\lambda \disc_{\lambda}(f)
\]
\end{definition}

Klauck \cite{Klauck07} proved that discrepancy and $\PP$ complexity are equivalent notions.
\begin{theorem}[Klauck \cite{Klauck07}]\label{thm: klauck}
For any function $f : \bra{-1, 1}^n \rightarrow \bra{-1, 1}$,
\[
\PP(f) = \Theta\left(\log\left(\frac{1}{\disc(f)}\right)\right)
\]
\end{theorem}
In \cite{GHR92}, Goldmann et al.~exhibited a distribution under which the one way communication complexity of $U_{4n, n} \circ \xor$ is large.  Sherstov \cite{She08} noted
that the same proof can be used to show that $\disc(U_{4n, n} \circ \xor) \leq O\left(\frac{\sqrt{n}}{2^{n/2}}\right)$.

\begin{remark}
We remark here that the function considered by Goldmann et al.~was not exactly $U_{4n, n} \circ \xor$, because the variables feeding to the $\xor$ gates had a mild dependence on each other.  Thus the discrepancy bound they obtained was slightly stronger than as stated above.
However, we will refer to $\UTHR \circ \xor$ as the $\GHR$ function.
\end{remark}

Sherstov defined the notion of a pattern matrix communication game in \cite{She11}.  Let $n$ be a positive integer and $f : \bra{0, 1}^n \rightarrow \bra{-1, 1}$.  Alice is given 2$n$ bits $x_{1, 1}, x_{1, 2}, x_{2, 1}, x_{2, 2}, \dots, x_{n, 1}, x_{n, 2}$.
Bob is given 2$n$ bits $z_1, z_2, \dots, z_n, w_1, w_2, \dots w_n$.  Define $\PM$ to be the function on 4 bits defined as $\PM(x_0, x_1, z, w) = x_z \oplus w$.
In the pattern matrix game corresponding to $f$, the $\PM$ gadget is applied on each tuple $\bra{x_{i, 1}, x_{i_2}, z_i, w_i}$, and the resultant $n$ bit string is fed as input to $f$.
This is the composed function, $f \circ \PM$.  Notice that this is similar to the lifting as defined in Equation \ref{eqn: select}.
\begin{theorem}[\cite{She11} Thm 1.5]\label{thm: pm}
Let $F = f \circ \PM$ for a given function $f : \bra{0, 1}^{n} \rightarrow \bra{-1, 1}$.
Then
\[
\disc(F) \leq \min_{d = 1, \dots, n}\max\left\{\left(\frac{n}{W(f, d - 1)}\right)^{1/2}, \left(\frac{1}{2}\right)^{d/2}\right\}
\]
\end{theorem}
In the above theorem, $W(f, d - 1)$ corresponds to the minimum weight of a polynomial of degree $d - 1$ with integer weights which sign represents $f$.

\begin{remark}
Sherstov defined pattern matrices in a more general fashion, where $n$ bits could be split into $t$ blocks containing $n/t$ elements each.  However, for the purposes of this paper, we only consider the case when each block is of size 2.
\end{remark}

The following theorem, first proposed by Klauck \cite{Klauck07}, provides a tool for proving bounded error communication lower bounds for functions.  Its proof may be found in \cite{Cha09, CA08}, for example.
\begin{theorem}[Generalized discrepancy]\label{thm: gendisc}
Let $F, G : \bra{-1, 1}^n \times \bra{-1, 1}^n \rightarrow \bra{-1, 1}$ and $\nu$ be a distribution over $\bra{-1, 1}^n \times \bra{-1, 1}^n$ such that $\corr_\nu(F, G) \geq \delta$.  Then.
\[
R_{\epsilon}(f) \geq \log\left(\frac{\delta - 1 + 2\epsilon}{\disc_\nu(G)}\right)
\]
\end{theorem}

For notational convenience, we use the notation $U(f)$ to represent $\UPP(f \circ \xor)$.
We also use the notation $U(\mod_m)$ to denote the minimum value of $U(\mod_m^A)$ over all non-simple accepting sets $A$.

Paturi and Simon \cite{PS86} showed an equivalence between $\UPP(f)$ and a quantity called the sign rank of $M_f$ where $M_f$ denotes the communication matrix of $f$.
Define the sign rank of a real matrix $M$ with no 0 entries as follows.
\begin{definition}[Sign Rank]
\[
sr(M) = \min_{A}\bra{rk(A) : sgn(A_{ij}) = sgn(M_{ij})}
\]
\end{definition}
We overload notation and use $sr(f)$ to denote $sr(M_f)$.
\begin{theorem}[Paturi and Simon \cite{PS86}]\label{thm: PS}
\[
\UPP(f) = \log sr(A) \pm O(1)
\]
\end{theorem}
Finding an explicit matrix with superlogarithmic sign rank remained a challenge until a breakthrough result of Forster \cite{Forster01}, who proved that the sign rank of any $N \times N$ Hadamard matrix is at least $\Omega(\sqrt{N})$.
This implied an asymptotically tight lower bound for the unbounded error communication complexity of the inner product (modulo 2) function.  We use a generalization of Forster's theorem by Forster et al.~\cite{FKLMSS01}.
\begin{theorem}[Forster et al.~\cite{FKLMSS01}]\label{thm: forster}
Let $M_{m \times N}$ be a real matrix with no 0 entries.  Then,
\[
sr(M) \geq \frac{\sqrt{mN}}{||M||} \cdot \min_{x, y}{\abs{M(x, y)}}
\]
where $||M||$ denotes the spectral norm of the matrix $M$.
\end{theorem}
Thus, it suffices to prove upper bounds on the spectral norm of the communication matrix of a function in order to prove unbounded error lower bounds for that function.
Let us now state a lemma characterizing the spectral norm of the communication matrix of \textsf{XOR} functions.
\begin{lemma}[Folklore]\label{lem: xor_eigenvalues}
Let $f: \bra{0, 1}^n \times \bra{0, 1}^n \rightarrow \mathbb{R}$ be any real valued function and let $M$ denote the communication matrix of $f \circ \xor$.  Then,
\[
||M|| = 2^n \cdot \max_{S \subseteq [n]}\abs{\wh{f}(S)}
\]
\end{lemma}
Although this is farily well known, we supply a proof below for completeness.
\begin{proof}
Let $M$ denote the communcation matrix of $f \circ \xor$.  That is, $M_{x, y} = f(x \oplus y)$.
Corresponding to each $T \subseteq [n]$, consider the vector $z_T \in \bra{-1, 1}^{2^n}$ defined by $(z_T)_y = \chi_T(y)$.

Note that
\begin{align*}
M_{x, y} = \sum_{S \subseteq [n]} \wh{f}(S) \chi_S(x \oplus y) = \sum_{S \subseteq [n]} \wh{f}(S) \chi_S(x) \chi_S(y)
\end{align*}
Fix any $T \subseteq [n]$.  We now show $z_T$ is an eigenvector of $M$ with eigenvalue $2^n\wh{f}(T)$.
Consider the $x$th coordinate of $Mz_T$.
\begin{align*}
(Mz_T)_x & = \sum_{y \in \bra{0, 1}^n}\sum_{S \subseteq [n]}\wh{f}(S)\chi_S(x)\chi_S(y)\chi_T(y)\\
& = \sum_S {\wh{f}(S) \chi_S(x)} \sum_{y \in \bra{0, 1}^n}\chi_{S \triangle T}(y)\\
& = \wh{f}(T)\chi_T(x)2^n
\end{align*}

Hence the eigenvalues of $M$ are precisely $\bra{2^n\wh{f}(S) : S \subseteq [n]}$.
Now, the singular values of $M$ are just the square root of the eigenvalues of $M^TM$, which are the absolute values of the eigenvalues of $M$ since $M$ is symmetric.
The lemma now follows.

\end{proof}

\section{Lifting functions}

In this section we first show how we `lift' functions as introduced by Krause and Pudl\'{a}k \cite{KP97}.  We then show how certain hardness properties of the base function
translate to related hardness properties of the lifted function.  Then, we show how lifted functions can be embedded in certain simple functions, if the base function was simple itself.
Finally, we list the consequences we obtain for lifting symmetric functions, which include resolving conjectures posed by Ada et al.~\cite{AFH12} and Zhang \cite{Zhang91}.

\subsection{Lifting functions by the Krause-Pudl\'{a}k selector}

In this section, we show how certain hardness properties of a function $f$ can be amplified into other hardness properties of a particular lifted function obtained from $f$.

For any $f : \bra{-1, 1}^n \rightarrow \bra{-1, 1}$, define a function $f^{op} : \bra{-1, 1}^{3n} \rightarrow \bra{-1, 1}$ as follows.
\begin{equation}\label{eqn: select}
f^{op}(x_1, \dots , x_n, y_1, \dots , y_n, z_1, \dots, z_n) = f(u_1, \dots, u_n)
\end{equation}
where for all $i$, $u_i = (x_i \wedge z_i) \vee (y_i \wedge \bar{z_i})$.
Intuitively speaking, the value of $z_i$ decides whether to feed $x_i$ or $y_i$ as the $i$th input to $f$.
This method of lifting $f$ was introduced by Krause and Pudl\'{a}k \cite{KP97}.

The following lemma translates hardness properties of $f$ into other hardness properties of $f^{op}$.  The proof of this lemma is based on ideas from \cite{KP97}.

\begin{lemma}\label{lem: lift}
Let $f : \bra{-1, 1}^n \rightarrow \bra{-1, 1}$ be any function.
\begin{enumerate}
\item\label{item: test} If $\varepsilon_d(f) > 1 - 2^{-d}$ for some $d \geq 2$, then $m(f^{op}) \leq 2^{-c'd}$ for any constant $0 < c' < 1 - \frac{1}{d}$.
\item $\mathrm{mon}_{\pm}(f^{op}) \geq 2^{\deg_{\pm}(f)}$.
\item $wt_{1/3}(f^{op}) \geq 2^{c \cdot \widetilde{\deg}_{2/3}(f)}$ for any constant $c < 1 - 1/\widetilde{\deg}_{2/3}(f)$.
\end{enumerate}

\end{lemma}

\begin{proof}
We first prove part \ref{item: test}.

Let $p$ be a polynomial of weight 1 representing $f^{op}$ with margin at least $\frac{1}{2^{c'd}}$ for a fixed positive constant $0 < c' < 1 - \frac{1}{d}$, and say $p = \sum\limits_{S \subseteq [n] \times [n] \times [n]}w_S \chi_S$.
Recall that $f^{op}$ (and also $p$) has $3n$ input variables.
For this proof, we view the input variables as $\bra{x_{j, 1}, x_{j, 2}, z_j | j \in \bra{1, \dots, n}}$, where $z_i$'s are the `selector' variables.

For any fixing of the $z$ variables, define a relevant variable to be one that is `selected' by $z$.
Thus, for each $j \in \bra{1, \dots, n}$, exactly one of $\bra{x_{j, 1}, x_{j, 2}}$ is relevant.
Analogously, define a relevant monomial to be one that contains only those variables selected by $z$.
For a uniformly random fixing of $z$ and any subset $S \subseteq [n]$ such that $|S| \geq d$,
\[
\Pr_{z}[\chi_S \text{ is relevant}] \leq \frac{1}{2^d}
\]
Now since $wt(p) = 1$, we have
\begin{align*}
\E_{z}[\text{weight of relevant monomials in }\restr{p}{z} \text{ of degree at least }d] & = \sum\limits_{|S| \geq d}|w_S| \cdot \Pr_z[\chi_S \text{ is relevant}]\\
& \leq \frac{1}{2^d}\sum\limits_{|S| \geq d}|w_S| \leq \frac{1}{2^d}
\end{align*}
Thus, there exists a fixing of the $z$ variables such that the weight of the relevant monomials of degree at least $d$ in $\restr{p}{z}$ is at most $\frac{1}{2^d}$.
Select this fixing of $z$.
\begin{itemize}
\item Note that $\restr{p}{z}$ is a polynomial on only the variables $\bra{x_{i, 1}, x_{i, 2} | i \in \bra{1, \dots, n}}$.  Drop the relevant monomials of degree at least $d$ from $\restr{p}{z}$ to obtain a polynomial $p_1$.
\item Observe that $p_1$ sign represents $\restr{f^{op}}{z}$ with margin at least $\frac{1}{2^{c'd}} - \frac{1}{2^d}$.
\item For each $j \in \bra{1, \dots, n}$, denote the irrelevant variable by $x_{j, i_j}$.
Consider the polynomial $p_2$ on $n$ variables defined by $p_2 = \E_{x_{1, i_1}, \dots, x_{n, i_n}}[p_1]$, where the expectation is over each irrelevant variable being sampled uniformly and independently from $\bra{-1, 1}$.
\item It is easy to see that any monomial containing an irrelevant variable in $p_1$ vanishes in $p_2$.  Also note that $p_2$ is a polynomial of degree at most $d$, and it must sign represent $f$ with margin at least $\frac{1}{2^{c'd}} - \frac{1}{2^d}$.
This leads to a contradiction when $c' < 1 - \frac{1}{d}$, since we assumed that $\varepsilon_d(f) > 1 - \frac{1}{2^d}$.
\end{itemize}

We omit the proofs of the other two statements as they follow along extremely similar lines.
\end{proof}

\subsection{Lifts as projections of simpler functions}

In this section, we show how lifts of threshold (and symmetric) functions can be viewed as the projections of threshold (symmetric) functions.

\begin{definition}[Monomial projection]\label{defn: monproj}
We call a function $g : \bra{-1, 1}^m \rightarrow \bra{-1, 1}$ a \emph{monomial projection} of a function $f : \bra{-1, 1}^n \rightarrow \bra{-1, 1}$ if $g(x_1, \dots, x_m) = f(M_1, \dots, M_n)$, where
each $M_i$ is a monomial in the variables $x_1, \dots, x_m$.
\end{definition}
The following lemma is an easy consequence of definitions.
\begin{lemma}\label{lem: monproj}
For any functions $f : \bra{-1, 1}^n \rightarrow \bra{-1, 1}$ and $g : \bra{-1, 1}^m \rightarrow \bra{-1, 1}$ such that $g$ is a monomial projection of $f$, and any $\epsilon > 0$, we have
\begin{align*}
m(f) & \leq m(g),\\
\mathrm{mon}_{\pm}(g) & \leq \mathrm{mon}_{\pm}(f),\\
wt(g) & \leq wt(f),\\
wt_\epsilon(g) & \leq wt_\epsilon(f).
\end{align*}
\end{lemma}

We first show that any lifted threshold function can be viewed as a monomial projection of a threshold function with a similar number of input variables.
This proof is based on methods of \cite{Krause06}.
\begin{lemma}\label{lem: thr_lift}
Given any linear threshold function $f : \bra{-1, 1}^n \rightarrow \bra{-1, 1}$, there exists a linear threshold function
$f' : \bra{-1, 1}^{4n} \rightarrow \bra{-1, 1}$ such that $f^{op}$ is a monomial projection of $f$.
\end{lemma}

\begin{proof}
Let $f : \bra{-1, 1}^n \rightarrow \bra{-1, 1}$ be a linear threshold function such that $m(f^{op}) \leq \delta$.
Fix a threshold representation for $f$, that is $f(x) = sgn\left(\sum\limits_{i = 1}^n w_i x_i\right)$.
Note that
\begin{align*}
f^{op}(x, y, z) & = sgn\left(\sum\limits_{i = 1}^n w_i\left(\frac{x_i(1 - z_i)}{2} + \frac{y_i(1 + z_i)}{2}\right)\right)\\
& = sgn\left(\sum\limits_{i = 1}^n w_i (x_i + y_i - x_iz_i + y_iz_i)\right)\\
\end{align*}

Consider a linear threshold function $f' : \bra{-1, 1}^{4n} \rightarrow \bra{-1, 1}$ defined as
\[
f'(x, y, u, v) = sgn\left(\sum\limits_{i = 1}^n w_i(x_i + y_i - u_i + v_i)\right)
\]

Clearly, $f^{op}$ is a monomial projection of $f'$.
\end{proof}

\begin{lemma}\label{lemma: symm_lift}
Given a symmetric function $F : \bra{-1, 1}^{4n} \rightarrow \bra{-1, 1}$, defined by the predicate $D_F : [n] \rightarrow \bra{-1, 1}$,
define a symmetric function $f : \bra{-1, 1}^{n} \rightarrow \bra{-1, 1}$ defined by the predicate $D_f(b) = D_F(2b + n)$ for all $b \in \bra{0, 1, \dots, n}$.
Then, $f^{op}$ is a monomial projection of $F$.  
\end{lemma}

\begin{proof}
Let $g : \bra{-1, 1}^{3n} \rightarrow \bra{-1, 1}$ be defined as follows.
\[
g(x_1, \dots x_n, y_1, \dots, y_n, z_1, \dots, z_n) = F(x_1, \dots, x_n, y_1, \dots, y_n, -x_1z_1, \dots, -x_nz_n, y_1z_1, \dots, y_nz_n).
\]
Clearly, $g$ is a monomial projection of $F$.  We show now that $g = f^{op}$.

For every input to $g$ and each $i \in [n]$, define the $i$'th \emph{relevant} variable to be $x_i$ if $z_i = -1$ (define $y_i$ to be the \emph{irrelevant} variable in this case), and $y_i$ if $z_1 = 1$ ($x_i$ is irrelevant in this case).
Suppose there are $b$ many relevant variables with value $-1$ on a fixed input $x_1, \dots, x_n, y_1, \dots, y_n, z_1, \dots, z_n$ and $n - b$ relevant variables with value $1$.
Say $(x_1, \dots, x_n, y_1, \dots, y_n, -x_1z_1, \dots, -x_nz_n, y_1z_1, \dots, y_nz_n)$ contains $a$ many $-1$'s.  Then,
\begin{align*}
4n - 2a & = \sum_{i = 1}^n x_i + y_i -x_iz_i + y_iz_i = \sum_{i = 1}^n x_i(1 - z_i) + y_i(1 + z_i) = 2n - 4b\\
& \implies a = 2b + n
\end{align*}
Thus, 
\begin{align*}
g(x_1, \dots, x_n, y_1, \dots, y_n, z_1, \dots, z_n) & = D_F(2b + n) = D_f(b)\\
& = f^{op}(x_1, \dots, x_n, y_1, \dots, y_n, z_1, \dots, z_n)
\end{align*}
The last equality follows from Equation \ref{eqn: select}.
\end{proof}

In fact, the proof of Lemma \ref{lemma: symm_lift} can be seen to imply the following lemma.
\begin{lemma}\label{lemma: lifsym}
Given a symmetric function $f : \bra{-1, 1}^{n} \rightarrow \bra{-1, 1}$ defined by the predicate $D_f(b)$, define a function $F : \bra{-1, 1}^{4n} \rightarrow \bra{-1, 1}$ such that on inputs of Hamming weight $2b + n$
for some $b \in \bra{0, 1, \dots, n}, ~ F$ takes the value $D_f(b)$, and $F$ takes arbitrary values on inputs of Hamming weight not in $\bra{2b + n : b \in \bra{0, 1, \dots, n}}$.  Then, $f^{op}$ is a monomial projection of $F$.
\end{lemma}

\subsection{Consequences for symmetric functions}\label{sec: consequences}

In this section, we show consequences of hardness amplification of lifted symmetric functions.

We first prove Theorem \ref{thm: liftsym}.

\begin{proof}[Proof of Theorem \ref{thm: liftsym}]\label{pf: listsym}
\hspace{2em}
\begin{itemize}
\item Assume that $n$ is even and that $r - 1$ is a multiple of 4.  (If not, we can fix a constant number of input bits).
Note that $D_F(r - 1) \neq D_F(r + 1)$.  Further assume $r_0(F) > r_1(F)$.
Define $F' : \bra{0, 1}^{2r} \rightarrow \bra{-1, 1}$ by $D_{F'}(i) = D_F(i)$.  It suffices to show $\log wt_{1/3}(F') \geq c'r$ for some universal constant $c' > 0$.
(If $r_1(F) \geq r_0(F)$, define $F' : \bra{0, 1}^{2r} \rightarrow \bra{-1, 1}$ by $D_{F'}(i) = D_F(4n - 2r + i)$, and an analogous argument to the one that follows can be carried out.
Define $f : \bra{0, 1}^{(r - 1)/2} \rightarrow \bra{-1, 1}$ by $D_f(i) = D_{F'}(2i + (r-1)/2)$.
By Lemma \ref{lemma: symm_lift}, $f^{op}$ is a monomial projection of $F'$.
Note that $D_f\left(\frac{r-1}{4}\right) \neq D_f\left(\frac{r-1}{4} + 1\right)$, and thus $\Gamma(f) \leq 1$.
By Theorem \ref{thm: paturi}, $\widetilde{\deg}_{2/3}(f) = \Theta(r)$.

Using Lemma \ref{lem: lift} and Lemma \ref{lem: monproj}, we obtain that there exists a universal constant $c_1 > 0$ such that
\begin{equation}\label{eqn: afhsolve}
\log(wt_{1/3}(F)) \geq \log(wt_{1/3}(F')) \geq \log(wt_{1/3}(f^{op})) \geq c_1r
\end{equation}
\item Consider any symmetric function $F : \bra{-1, 1}^{4n} \rightarrow \bra{-1, 1}$ such that $\deg_{oe}(F) \geq 4j$ where $j \geq 4$.
Assume that there are at least $2j$ many $(i, i + 2)$ sign changes in $[0, 3n]$.
Further assume that at least $j$ of them occur when $i$'s are even integers (if not, set one variable to $-1$).
Define a family of symmetric functions $\bra{f_i : \bra{-1, 1}^{\frac{4n}{3^i}} \rightarrow \bra{-1, 1} : i \in \bra{0, 1, \dots, \lceil \frac{1}{\log 3}\log\left(\frac{2n}{j}\right)\rceil}}$ as follows.
\[
\forall b \in \left[\frac{4n}{3^i}\right], ~ D_{f_i}(b) = D_F\left(2b + \frac{4n}{3^i}\right).
\]
(If there were less than $j$ many $(i, i + 2)$ sign changes in $[0, 3n]$ for even integers $i$, then there must be at least $j$ many $(i, i + 2)$ sign changes in $[n, 4n]$.  In this case, define $D_{f_i}(b) = D_F\left(4n - 2b - \frac{4n}{3^i}\right)$,
and an argument similar to the one that follows can be carried out).

Note that the sign degree of $f_i$ equals the number of $(k, k + 2)$ sign changes in the spectrum of $F$ in the interval $[\frac{n}{3^i}, \frac{n}{3^{i - 1}}]$.
Since $D_F$ has at least $\lfloor j/2\rfloor$ many $(k, k + 2)$ sign changes in the interval $[\lfloor j/2\rfloor, 3n]$, this implies that at least one of the $f_i$'s has at least 
$\frac{\lfloor j \rfloor/2}{\lceil \frac{1}{\log 3}\log\left(\frac{2n}{j}\right)\rceil}$ many $(k, k + 1)$ sign changes (sign degree).
Using Lemma \ref{lemma: symm_lift}, Lemma \ref{lem: lift} and Lemma \ref{lem: monproj}, we obtain that there exists a constant $c_2 > 0$ such that
\[
\mathrm{mon}_{\pm}(F) \geq 2^{c_2j}.
\]
\item The proof of the Part 3 follows along extremely similar lines as that of Part 2, and we omit it.
\end{itemize}
\end{proof}

We next prove Theorem \ref{thm: afhsolve}, resolving a conjecture of Ada et al.~\cite{AFH12}.
\begin{proof}[Proof of Theorem \ref{thm: afhsolve}]
It follows as a direct consequence of Part 1 of Theorem \ref{thm: liftsym} and the upper bound in Theorem \ref{thm: afh}.
\end{proof}

Finally, we prove Theorem \ref{thm: zhangsolve} here, settling a conjecture of Zhang \cite{Zhang91}.
\begin{proof}[Proof of Theorem \ref{thm: zhangsolve}]

The upper bound follows from Theorem \ref{thm: zhang_ub}.  It suffices to show a lower bound for when $\deg_{oe}(f) \geq 16$.
The lower bound follows from Part 2 of Theorem \ref{thm: liftsym} in this case.

\end{proof}

\section{Discrepancy of \textsf{XOR} functions}
In this section, we analyze the discrepancy of $\xor$ functions.

\subsection{Margin-discrepancy equivalence}\label{sec: equiv}

In this section, we prove Theorem \ref{thm: equiv}, which is a necessary and sufficient approximation theoretic condition of $f$ in order for $f \circ \xor$ to have small discrepancy.

\begin{proof}[Proof of Theorem \ref{thm: equiv}]

We first show that $m(f) \leq m(f \circ \xor)$.  For notational convenience, let us denote $f \circ \xor$ by $F$.
View $f$'s inputs as $x_1, \dots x_n$, and $F$'s inputs as $y_1, \dots, y_n, z_1, \dots, z_n$, where $f$ is fed $y_1 \oplus z_1, \dots, y_n \oplus z_n$.
Let $p$ be any polynomial of weight 1 sign representing $f$.
Replace every variable $x_i$ in $p$ by $y_iz_i$.  Clearly, the new polynomial obtained sign represents $F$ with the same margin as $p$ represented $f$, and the weight remains unchanged.
Thus, $m(f) \leq m(F)$.

Next, we show that $m(F) \leq 4\disc(F)$.
Let $\lambda$ denote a distribution under which $\disc_\lambda(F) = \disc(F)$, and let $P(x, y) = \sum_{S \subseteq [2n]}c_S \chi_S(x, y)$ be a polynomial of weight 1, which sign represents $F$.

\begin{align*}
m(F) & \leq \E_{\lambda}[F(x, y)P(x, y)]\\
& \leq \E_{\lambda}\left[F(x, y)\sum_{S \subseteq [2n]}c_S \chi_S(x, y)\right]\\
& \leq \left(\sum_{S \subseteq [2n]} \abs{c_S}\right) \cdot \max_{S \subseteq [2n]} \left(\abs{\E_{\lambda}[F(x, y) \chi_S(x, y)]}\right)\\
& \leq \abs{\sum_{\substack{\chi_S(x) = 1 \\ \chi_S(y) = 1}} F(x, y)\lambda(x, y)} + \abs{\sum_{\substack{\chi_S(x) = 1 \\ \chi_S(y) = -1}} F(x, y)\lambda(x, y)} + \abs{\sum_{\substack{\chi_S(x) = -1 \\ \chi_S(y) = 1}} F(x, y)\lambda(x, y)}\\
& + \abs{\sum_{\substack{\chi_S(x) = -1 \\ \chi_S(y) = -1}} F(x, y)\lambda(x, y)}\\
& \leq 4\disc(F)
\end{align*}

Thus, $m(F) \leq 4\disc(F)$.

Now we show that $\disc(F) \leq m(f)$.

Let us first write a linear program whose optimal value corresponds to the margin of $f$.

\begin{center}
\begin{framed}
\label{lp: primal}
\begin{tabular}{llllll}
Variables & $\Delta, \{\alpha_S: S \subseteq [n]\}$ &               &  &                             &  \\
Maximize  & $\Delta$                                &               &  &                             &  \\
s.t.      & $f(x) \sum\limits_{S \subseteq [n]}\alpha_S\chi_S(x)$ & $\geq \Delta$ &  & $\forall x \in \{-1, 1\}^n$ &  \\
          & $\sum\limits_{S \subseteq [n]}|\alpha_S|$             &  $\leq 1$             &  &                             &  \\
          & $\Delta \in \mathbb{R}$                                   &              &  &                             &  \\
          & $\alpha_S \in \mathbb{R}$                                        &               &  &$\forall S \subseteq [n]$                             & 
\end{tabular}
\end{framed}
\end{center}

We write another linear program, which is easier to work with.

\begin{center}
\begin{framed}
\label{lp: primal'}
\begin{tabular}{llllll}
Variables & $\Delta, \{\alpha'_{S}: S \subseteq [n]\}, \bra{\alpha''_{S} : S \subseteq [n]}$ &               &  &                             &  \\
Maximize  & $\Delta$                                &               &  &                             &  \\
s.t.      & $f(x) \sum\limits_{S \subseteq [n]}\chi_S(x)(\alpha''_S - \alpha'_S)$ & $\geq \Delta$ &  & $\forall x \in \{-1, 1\}^n$ &  \\
          & $\sum\limits_{S \subseteq [n]}(\alpha'_S + \alpha''_S)$             &  $\leq 1$             &  &                             &  \\
          & $\Delta \in \mathbb{R}$                                   &              &  &                             &  \\
          & $\alpha'_S, \alpha''_S \geq 0$                                        &               &  &$\forall S \subseteq [n]$                             & 
\end{tabular}
\end{framed}
\end{center}

Note that any solution to the first program is a valid solution to the second one, by setting one of $\alpha'_S$ or $\alpha''_S$ to 0, and the other to $\abs{\alpha_S}$ for each $S \subseteq [n]$.
We can also assume that a solution to the second program must have $\alpha'_S = 0$ or $\alpha''_S = 0$ for each $S \subseteq [n]$.  If this was not the case, one could reduce the values of $\alpha'_S$ and $\alpha''_S$
by the same amount, thus not changing the value of $\alpha''_S - \alpha'_S$, and not violating any constraints.
This gives us a solution to the first program by setting $\alpha_S = \alpha''_S$ if $\alpha''_S \neq 0$, and $\alpha_S = \alpha'_S$ otherwise.
Thus, the optima of the two programs above are equal.

Let us now look at the corresponding dual to the above linear program.  Notice that the program looks like a minimization problem with the objective to minimize $\max\limits_{S \subseteq [n]}{\abs{\wh{f\mu}(S)}}$ under a variable distribution $\mu$ on $\bra{-1, 1}^n$.

\begin{center}
\begin{framed}
\label{lp: dual}
\begin{tabular}{llllll}
Variables & $\epsilon, \bra{\mu(x): x \in \bra{-1, 1}^n}$ &               &  &                             &  \\
Minimize  & $\epsilon$                                &               &  &                             &  \\
s.t.      & $|\sum\limits_{x}\mu(x)f(x)\chi_S(x)|$ & $\leq \epsilon$ &  & $\forall S \subseteq [n]$ &  \\
          & $\sum\limits_{x}\mu(x)$             &  $= 1$             &  &                             &  \\
          & $\epsilon \geq 0$                                   &              &  &                             &  \\
          & $\mu(x) \geq 0$                                        &               &  &$\forall x \in \bra{-1, 1}^n$                             & 
\end{tabular}
\end{framed}
\end{center}

Thus, if $f$ has margin at most $\delta$, there exists a distribution $\mu$ on $\bra{-1, 1}^n$ such that $\abs{\wh{f\mu}(S)} \leq \frac{\delta}{2^n}$ for all $S \subseteq [n]$.
Let $\mu^{\oplus}$ be a distribution denoting the lift of $\mu$ on $\bra{-1, 1}^n \times \bra{-1, 1}^n$.  That is, $\mu^\oplus(x, y) = \frac{1}{2^n}\mu(x \oplus y)$.
We now show that the discrepancy of $F$ is small under $\mu^\oplus$.
For matrices $A, B$, let $A \circ_H B$ denote the Hadamard (entrywise) product of $A$ and $B$.
Note that under the distribution $\mu^\oplus$, the discrepancy of $F$ is
\begin{align*}
\disc_{\mu^\oplus}(F) & = \max_{S \subseteq [n], T \subseteq [n]} \mathsf{1}_S^T (\mu^\oplus \circ_H F) \mathsf{1}_T\\
& \leq ||\mu^\oplus \circ_H F|| \cdot 2^n \tag*{Cauchy-Schwarz}
\end{align*}
Thus,
\[
\disc_{\mu^\oplus}(F) \leq \frac{||f\mu \circ \xor||}{2^n} \cdot 2^n = 2^n \cdot ||\wh{f\mu}||_{\infty} \leq \delta
\]

Here, the first inequality follows from the definition of $\mu^\oplus$, and the following equality follows from Lemma \ref{lem: xor_eigenvalues}.
This proves the claim.

\end{proof}

We remark here that Linial and Shraibman \cite{LS09a} had shown a similar equivalence between the discrepancy of a matrix (the communication matrix of the target function) and its margin.
This margin refers to the margin of the matrix, and not the base function.  However, since we do not use this notion in the rest of this paper, we overload notation and use $m(A)$ to denote the margin of the matrix $A$.
Define the margin of an $m \times n$ sign matrix $A$ as follows.
\[
m(A) = \sup \min_{i, j}\frac{\abs{\langle x_i, y_j \rangle}}{||x_i||_2 ||y_j||_2}
\]
where the supremum is over all choices of $x_1, \dots, x_m, y_1, \dots, y_n \in \R^{m + n}$ such that $sgn(\langle x_i, y_j \rangle) = a_{i, j}$ for all $i, j$.
Linial and Shraibman \cite{LS09a} showed that the margin of a sign matrix is equivalent to its discrepancy up to a constant factor.

\begin{theorem}[\cite{LS09a} Thm 3.1]\label{thm: LS}
For every sign matrix $A$,
\[
\disc(A) \leq m(A) \leq 8\disc(A)
\]
\end{theorem}

We now note that Theorem \ref{thm: equiv} implies the first inequality of Theorem \ref{thm: LS} for the special case of $\xor$ functions.
\begin{claim}\label{claim: LS_implied}
Let $f : \bra{-1, 1}^n \rightarrow \bra{-1, 1}$.  Then,
\[
m(f) \leq m(M_{f \circ \xor})
\]
\end{claim}
\begin{proof}
Let $p = \sum_{S \subseteq [n]}c_S\chi_S$ be a polynomial which sign represents $f$ with margin $\delta$.
This implies $p'(x, y) = \sum_{S \subseteq [n]} c_S\chi_S(x)\chi_S(y)$ sign represents $f \circ \xor$ with margin $\delta$.

We will exhibit $2^{n+1}$ vectors, $\bra{u_T: T \subseteq [n]}$ and $\bra{v_T: T \subseteq [n]}$ in $\R^{2^n}$ such that $m(M_{f \circ \xor}) \geq \delta$.
Index the coordinates by characteristic sets, $T \subseteq [n]$.  For a set $T \subseteq [n]$, we use $w_T$ to denote the corresponding characteristic vector in $\R^{2^n}$.
Define $u_T(S) = v_T(S) = \sqrt{c_S}\chi_S(w_T)$

Since $wt(p') = 1$, $||u_T||_2 = ||v_T||_2 = 1$.
Also, $\langle u_{T_1}, v_{T_2} \rangle = \sum_{S \subseteq [n]} c_S \chi_S(w_{T_1} \oplus w_{T_2}) \geq \delta$
since $p'$ sign represents $f \circ \xor$ with margin $\delta$.

Thus,
\[
m(M_{f \circ \xor}) = \sup \min_{T_1, T_2}\frac{\abs{\langle u_{T_1}, v_{T_2} \rangle}}{||u_{T_1}||_2 ||v_{T_2}||_2} \geq \delta
\]
\end{proof}

\subsection{A new separation of $\PP^{cc}$ from $\UPP^{cc}$}

In this section, we show here how to obtain an alternate proof that the $\GHR$ function has large $\PP$ complexity.  It is well known that $\GHR \in \UPP^{cc}$.

\begin{proof}[Proof of Theorem \ref{thm: main}]
Theorem \ref{thm: sherstov} and Lemma \ref{lem: lift} show the existence of a linear threshold function $f : \bra{-1, 1}^n \rightarrow \bra{-1, 1}$ such that $m(f^{op}) \leq 2^{-cn}$ for some absolute constant $c > 0$.
Lemma \ref{lem: monproj} and Lemma \ref{lem: thr_lift} then show existence of a linear threshold function $f' : \bra{-1, 1}^{4n} \rightarrow \bra{-1, 1}$ such that $m(f') \leq 2^{-cn}$.
Using Theorem \ref{thm: equiv} and Theorem \ref{thm: klauck}, we already obtain the existence of a linear threshold function $f' : \bra{-1, 1}^{4n} \rightarrow \bra{-1, 1}$ such that $\PP(f' \circ \xor) \geq c'n$ for some absolute constant $c' > 0$.

By Fact \ref{fact: quad_blowup}, one can embed $f'$ in the universal threshold function by blowing up the number of variables by a quadratic factor (note that we do not lose a logarithmic factor as stated in Fact \ref{fact: quad_blowup},
because it can be verified that the weights of $f'$ are at most $2^{\alpha n}$ for an absolute constant $\alpha > 0$).
Thus, $m(\UTHR) \leq 2^{-\Omega(\sqrt{n})}$.
By Theorem \ref{thm: equiv} and Theorem \ref{thm: klauck}, we have
\[
\PP(\GHR) \geq \Omega(\sqrt{n})
\]
\end{proof}

\subsection{$\xor$ is harder than $\PM$}

In this section, we observe that if $f \circ \xor$ has small discrepancy, then so does $f \circ \PM$.
Note that the converse is not true, since the inner product function is a large subfunction of $\oplus \circ \PM$, which has inverse exponential discrepancy, but $\oplus \circ \xor$ has extremely large discrepancy.

\begin{theorem}\label{thm: xorand}
Let $f : \bra{-1, 1}^n \rightarrow \bra{-1, 1}$.  Then,
\[
\disc(f \circ \xor) < \delta \implies \disc(f \circ \PM) \leq \sqrt{4 \delta n}
\]
\end{theorem}
\begin{proof}
Consider $f \circ \PM$ and substitute $d = n$ in Theorem \ref{thm: pm} to obtain
\[
\disc(f \circ \PM) \leq \left(\frac{n}{W(f, d - 1)}\right)^{1/2}
\]
By Theorem \ref{thm: equiv}, $\disc(f \circ \xor) < \delta \implies m(f) < 4\delta$.
Suppose $W(f, n - 1) \leq \frac{1}{4\delta}$.  This would show existence of a polynomial with integer weights, say $\sum_{S \subseteq [n]} \lambda_S\chi_S$, sign representing $f$, and with total weight at most $1/4\delta$.
This in turn implies existence of a polynomial of weight 1, $p = \frac{\sum_{S \subseteq [n]} \lambda_S\chi_S}{\sum_{S \subseteq [n]} \abs{\lambda_S}}$,
which sign represents $f$ with margin at least $4\delta$, which is a contradiction.
Thus,
\[
\disc(f \circ \PM) \leq \sqrt{4\delta n}
\]

\end{proof}

\subsection{Symmetric functions with large odd-even degree}\label{sec: oddeven}

We show that for any symmetric function $F$, $\PP(F \circ \xor)$ is lower bounded by $\deg_{oe}(F)$ (up to a logarithmic factor in the input size).

\begin{proof}[Proof of Theorem \ref{thm: sym_xor}]
Using Theorem \ref{thm: equiv} and Part 3 of Theorem \ref{thm: liftsym}, we obtain that there exists a universal constant $c > 0$ such that $\PP(F \circ \xor) \geq cr/\log(n/r)$, which proves Theorem \ref{thm: sym_xor}.
\end{proof}

\section{Bounded error communication complexity of $\xor$ functions}\label{sec: bdd}

In this section, we analyze the bounded error communication complexity of $\xor$ functions.

\begin{proof}[Proof of Theorem \ref{thm: bpp}]
We write a linear program which captures the best error a weight $w$ polynomial can achieve in approximating a given function $f$.

\begin{center}
\begin{framed}
\label{lp: primalapprox}
\begin{tabular}{llllll}
Variables & $\epsilon, \{\alpha_S: S \subseteq [n]\}$ &               &  &                             &  \\
Minimize  & $\epsilon$                                &               &  &                             &  \\
s.t.      & $\abs{f(x) - \sum\limits_{S \subseteq [n]}\alpha_S\chi_S(x)}$ & $\leq \epsilon$ &  & $\forall x \in \{-1, 1\}^n$ &  \\
          & $\sum\limits_{S \subseteq [n]}|\alpha_S|$             &  $\leq w$             &  &                             &  \\
          & $\epsilon \geq 0$                                   &              &  &                             &  \\
          & $\alpha_S \in \mathbb{R}$                                        &               &  &$\forall S \subseteq [n]$                             & 
\end{tabular}
\end{framed}
\end{center}
By manipulations similar to those in Section \ref{sec: equiv}, we obtain the following dual program.

\begin{center}
\begin{framed}
\label{lp: dualapprox}
\begin{tabular}{llllll}
Variables & $\Delta, \bra{\mu(x): x \in \bra{-1, 1}^n}$ &               &  &                             &  \\
Maximize  & $\sum_x f(x)\mu(x) - \Delta w$                                &               &  &                             &  \\
s.t.      & $|\sum\limits_{x}\mu(x)\chi_S(x)|$ & $\leq \Delta$ &  & $\forall S \subseteq [n]$ &  \\
          & $\sum\limits_{x}\mu(x)$             &  $\leq 1$             &  &                             &  \\
          & $\Delta \geq 0$                                   &              &  &                             &  \\
          & $\mu(x) \geq 0$                                        &               &  &$\forall x \in \bra{-1, 1}^n$                             & 
\end{tabular}
\end{framed}
\end{center}

By strong linear programming duality, the optima of the two programs above are equal.  Let us call the optimal value \textsf{OPT}, which is clearly non-negative.
Note that in any feasible solution to the dual, $1 - \Delta w \geq \sum\limits_{x}f(x)\mu(x) - \Delta w \geq 0$.  This implies $\Delta \leq \frac{1}{w}$.
Suppose a function $f : \bra{-1, 1}^n \rightarrow \bra{-1, 1}$ satisfied $wt_{1/3}(f) = w'$.
This means if we fix $w = w'$ in the programs, then $\textsf{OPT} = 1/3$, which implies $\sum\limits_{x}f(x)\mu(x) \geq 1/3$ since $\Delta$ is non-negative.
Thus, any optimum solution to the dual must satisfy $\sum\limits_{x}\mu(x) \geq 1/3$.
Define a distribution $\mu'$ by $\mu'(x) = \frac{\mu(x)}{\sum_{x \in \bra{-1, 1}^n}\mu(x)}$, and we obtain $|\sum\limits_{x}\mu'(x)\chi_S(x)| \leq \frac{3}{w'}$ (hence, setting $\Delta = \frac{3}{w'}$ gives us a feasible solution).

Write $\mu' = g \cdot \nu$ uniquely, where $g : \bra{-1, 1}^n \rightarrow \bra{-1, 1}$ is a boolean function and $\nu : \bra{-1, 1}^n \rightarrow [0, 1]$ is a distribution on the inputs.
Thus, $\corr_\nu(f, g) \geq 1/3$ (which implies $\corr_{\nu^\oplus}(f \circ \xor, g \circ \xor) \geq 1/3$), and
\[
\disc_{\nu^\oplus}(g \circ \xor) \leq \frac{||g\nu \circ \xor||}{2^n} \cdot 2^n = 2^n \cdot ||\wh{g\nu}||_{\infty} \leq \Delta \leq \frac{3}{w'}.
\]
This, along with Theorem \ref{thm: gendisc} proves the following.
\[
R_{7/15}(f \circ \xor) \geq \log w' - 4
\]
By standard error reduction, we obtain Theorem \ref{thm: bpp}.
\end{proof}

Using Part 1 of Theorem \ref{thm: liftsym} and Theorem \ref{thm: bpp}, we obtain a new proof of Theorem \ref{thm: SZ}.

\section{Sign rank of \textsf{XOR} functions}\label{sec: sr}

In this section, we analyze the unbounded error communication complexity of $\xor$ functions.

\subsection{Fourier analysis of some modular functions}\label{subsec: fourier}

We first closely analyze the Fourier coefficients of functions of the type $\mod_m^A$, when $m$ is odd, using exponential sums.
\begin{claim}\label{claim: odd_fourier_coefficients}
For odd $m$, and any $A \subseteq \bra{0, 1, \dots, m - 1}$ which is not the full or empty set,
\[
\abs{\wh{\mod_m^A}(S)} \leq \begin{cases}
                   1 - \frac{2}{m} + 2m \left(\cos\left(\frac{\pi}{2m}\right)\right)^{n} & S = \emptyset\\
                   2m \left(\cos\left(\frac{\pi}{2m}\right)\right)^{n} & S \neq \emptyset
                  \end{cases}
\]
\end{claim}

Zhang \cite{Zhang91} showed that for a fixed prime $p$, $\abs{\wh{\mod_p}^{\bra{0}}(\emptyset)} < 1 - \frac{1}{p}$, and $\abs{\wh{\mod_p}^{\bra{0}}(S)} = O\left(\frac{1}{2^{\Omega(n)}}\right)$ when $S \neq \emptyset$.
We show that a similar bound holds for odd integers $m$ for values up to $m = O(n^{1/2 - \epsilon})$ using a different technique.
In particular, we show that for $m = O(n^{1/2 - \epsilon})$, the principal coefficient is roughly $1 - \frac{1}{m}$,
and all other coefficients are exponentially small $\left(\frac{1}{2^{n^{\Omega(1)}}}\right)$, for any non simple accepting set $A$.

We use the characterization of the $\mod_m^A$ function in terms of exponential sums to analyze its Fourier coefficients.
Note that exponential sums have been used in similar contexts in previous papers as well.  For example, the reader may refer to \cite{Bourgain05, CGPT06, ACFN15}.
The notation we use is that from \cite{ACFN15}.

\begin{definition}\label{defn: exp_sum}
Let $\omega = e^{2\pi i/m}$ be a primitive $m$-th root of unity.  Then, for $x = \bra{0, 1}^n$, define
\[
\exp^{a, b}_m(x_1, \dots, x_n) = \omega^{a\left(\left(\sum_{j = 1}^n x_j\right) - b\right)}
\]
\end{definition}
Let us now prove Claim \ref{claim: odd_fourier_coefficients}.
\begin{proof}
First, we use exponential sums to represent a $\mod_m^{A}$ function for odd $m$.

It is easy to check that for any integer $k$, and any input $x = (x_1, \dots, x_n)$,
\[
\frac{1}{m} \sum\limits_{a = 0}^{m - 1} \exp^{a, k}_m (x) = 
\begin{cases}
                                                                                                              1 & |x| \equiv k ~(\text{mod}~m)\\
                                                                                                                 0 & \text{otherwise}
                                                                                                                \end{cases}
\]
Thus, for a general accepting set $A \subseteq [m]$,
\[
\sum\limits_{k \in A}\left(\frac{1}{m} \sum\limits_{a = 0}^{m - 1} \exp^{a, k}_m (x)\right) = 
\begin{cases}
                                                                                                              1 & |x| \equiv k ~(\text{mod}~m) \text{ for some } k \in A\\
                                                                                                                 0 & \text{otherwise}
                                                                                                                \end{cases}
\]

Just by a simple linear transformation from $\bra{0, 1}$ to $\bra{-1, 1}$, we can express the $\mod_m^A$ function in terms of exponential sums as follows.
\begin{equation}\label{eqn: exp_sum}
\mod_m^A(x) = 1 - \frac{2}{m}\sum\limits_{k \in A}\left(\sum\limits_{a = 0}^{m - 1} \exp^{a, k}_m (x)\right) = \begin{cases}
                                                                                                                 -1 & |x| \equiv k (\text{mod}~m) ~\text{for some $k \in A$}\\
                                                                                                                 1 & \text{otherwise}
                                                                                                                \end{cases}
\end{equation}
Let us now look at the Fourier coefficients of $\mod_m^A$ for odd $m$, and $A$ not $\emptyset$ or $[m]$.  Let us consider 2 cases, the first where $S$ is non-empty, and the second where $S$ is empty.
\begin{enumerate}
\item $S \neq \emptyset$.\\
By Equation \eqref{eqn: fourier_coefficients},
\begin{align}
\wh{\mod_m^A}(S) & =\E_{x \in \bra{0, 1}^n} \left[\mod_m^A(x) \chi_S(x)\right] \nonumber \\
& = \E_{x \in \bra{0, 1}^n}\left[\chi_S(x)\right] - \frac{2}{m}\sum\limits_{k \in A} \sum\limits_{a = 0}^{m - 1}\E_{x \in \bra{0, 1}^n}\left[\exp_m^{a, k}(x)\chi_S(x)\right] \label{eqn: mod_fourier_nonempty}
\end{align}
where the second equality follows from Equation \eqref{eqn: exp_sum} and linearity of expectation.
Recall from Definition \ref{defn: exp_sum} that $\exp^{a, b}_m(x) = \omega^{a\left(\left(\sum_{j = 1}^n x_j\right) - b\right)}$.  Note that when $a = 0$, $\E_{x \in \bra{0, 1}^n}\left[\exp^{0, b}_m(x) \chi_S(x)\right] = \E_{x \in \bra{0, 1}^n}\left[\chi_S(x)\right] = 0$ since $S \neq \emptyset$.
For $a \in \bra{1, \dots, m-1}$,
\begin{align*}
\exp_m^{a, k}(x)\chi_S(x) & = \omega^{a\left(\left(\sum_{j = 1}^n x_j\right) - k\right)} (-1)^{\sum_{i \in S}x_i}\\
& = \omega^{a\sum_{j = 1}^nx_j} \cdot \omega^{-ak} \cdot (-1)^{\sum_{i \in S}x_i}\\
& = \omega^{-ak} \cdot (-\omega)^{a\sum_{i \in S}x_i} \cdot \omega^{a\sum_{j \notin S}x_j}
\end{align*}
Thus, in Equation \eqref{eqn: mod_fourier_nonempty}, the first term is 0 since $S \neq \emptyset$.  The summands with $a = 0$ contribute 0 to the expectation.  Every other summand in the second term is of the form $\E_{x \in \bra{0, 1}^n}\left[\exp_m^{a, k}(x)\chi_S(x)\right]$.  Since the expectation is over the uniform distribution which is uniform and independent over the input bits, the absolute value of such a term can be bounded as follows.
\begin{align*}
\abs{\E_{x \in \bra{0, 1}^n}\left[\exp_m^{a, k}(x)\chi_S(x)\right]} & \leq \abs{\E_{x \in \bra{0, 1}^n}\left[\omega^{-ak} \cdot (-\omega)^{a\sum_{i \in S}x_i} \cdot \omega^{a\sum_{j \notin S}x_j}\right]}\\
& \leq \abs{\prod_{i \in S}\E_{x_i}(-\omega)^{ax_i}} \cdot \abs{\prod_{j \notin S}\E_{x_j}\omega^{ax_j}}\\
& \leq \abs{\left(\frac{1 - \omega^a}{2}\right)}^{\abs{S}} \abs{\left(\frac{1 + \omega^a}{2}\right)}^{n - \abs{S}}\\
& \leq \max_{a \in \bra{1, \dots, m - 1}}\left\{\abs{\frac{1 - \omega^a}{2}}^n, \abs{\frac{1 + \omega^a}{2}}^n\right\}
\end{align*}

Since $a \in \bra{1, \dots, m - 1}$ and $m$ is odd, it is fairly straightforward to check that the value of $\max_a\left\{\abs{\frac{1 - \omega^a}{2}}, \abs{\frac{1 + \omega^a}{2}}\right\}$ is maximized at $a = \frac{m \pm 1}{2}$, and the value attained at the maximum is $\frac{1}{2}\sqrt{{(1 + \cos(\pi/m))^2 + \sin^2(\pi/m)}} = \frac{1}{2}\sqrt{2 + 2\cos(\pi/m)} = \cos(\pi/2m)$.
Thus, the above, along with Equation \eqref{eqn: mod_fourier_nonempty} gives us
\begin{align}
\abs{\wh{\mod_m^A}(S)} & \leq \abs{\E_{x \in \bra{0, 1}^n}\left[\chi_S(x)\right]} + \abs{\frac{2}{m}\sum\limits_{k \in A} \sum\limits_{a = 0}^{m - 1}\E_{x \in \bra{0, 1}^n}\left[\exp_m^{a, k}(x)\chi_S(x)\right]}\\
& \leq \frac{2(m - 1)^2}{m} \cdot \left(\cos\left(\frac{\pi}{2m}\right)\right)^{n} \leq 2m \left(\cos\left(\frac{\pi}{2m}\right)\right)^{n}
\end{align}
\item $S = \emptyset$.

One can follow a similar argument as above to analyze the absolute value of the principal Fourier coefficient.  Note that in this case, the first term on the right hand side of Equation \eqref{eqn: mod_fourier_nonempty} is not 0, but 1.
Next, note that for $a \in \bra{1, \dots, m - 1}$, the same bound as in the previous case holds.  That is,
\begin{align*}
\abs{\E_{x \in \bra{0, 1^n}}\left[\exp_m^{a, k}(x)\chi_S(x)\right]} & \leq \prod_{i \in S}\E_{x_i}(-\omega)^{ax_i} \cdot \prod_{j \notin S}\E_{x_j}\omega^{ax_j}\\
& \leq \abs{\left(\frac{1 - \omega^a}{2}\right)}^{\abs{S}} \cdot \abs{\left(\frac{1 + \omega^a}{2}\right)}^{n - \abs{S}}\\
& \leq \left(\cos\left(\frac{\pi}{2m}\right)\right)^{n}
\end{align*}
by the same argument as in the case of $S \neq \emptyset$.
However, when $S = \emptyset$ and $a = 0$, we have $\E_{x \in \bra{0, 1}^n}\left[\exp^{a, b}_m(x) \chi_\emptyset(x)\right] = 1$ (unlike the case when $S \neq \emptyset$, where this expectation was 0).

Plugging these values into Equation \eqref{eqn: mod_fourier_nonempty} and using the above observations, we get
\begin{align}
\abs{\wh{\mod_m^A}(\emptyset)} & \leq \abs{\E_{x \in \bra{0, 1}^n}\left[\chi_\emptyset(x)\right] - \frac{2}{m}\sum\limits_{k \in A}\E_{x \in \bra{0, 1}^n}\left[\exp_m^{0, k}(x)\chi_\emptyset(x)\right]} \nonumber\\
& + \abs{\frac{2}{m}\sum\limits_{k \in A} \sum\limits_{a = 1}^{m - 1}\E_{x \in \bra{0, 1}^n}\left[\exp_m^{a, k}(x)\chi_\emptyset(x)\right]}\\
& \leq \abs{1 - 2\frac{|A|}{m}} + 2m \left(\cos\left(\frac{\pi}{2m}\right)\right)^{n}\\
& \leq 1 - \frac{2}{m} + 2m \left(\cos\left(\frac{\pi}{2m}\right)\right)^{n} \tag*{since $A \neq \emptyset, [m]$}
\end{align}
\end{enumerate}
\end{proof}

\subsection{A lower bound for $\mod_m^A \circ \xor$}

In this section, we show unbounded error lower bounds for functions of the type $\mod_m^A \circ \xor$ for values of $m$ up to $O(n^{1/2 - \epsilon})$, when $A$ is non-simple.
Note that if $A$ is a simple set, then either $\mod_m^A \circ \xor$ is a constant or $\mod_m^A$ represents parity (or its negation), in which case $\mod_m^A \circ \xor$ just represents the parity function (or its negation),
so its communication complexity (even deterministic) is very small. We prove a new sign rank lower bound criterion for $\xor$ functions.
As an application of this theorem, we show that $\UPP(\mod_m \circ \xor) = n^{\Omega(1)}$ for values of odd $m$ up to $O(n^{1/2 - \epsilon})$.
Theorem \ref{thm: PS} tells us that the log of the sign rank of a communication matrix is essentially equivalent to the unbounded error communication complexity of the function.

Let $f: \bra{0, 1}^n \rightarrow \mathbb{R}$, and let $A$ denote the communication matrix of $f \circ \xor$.
In order to show a lower bound on the sign rank of $f \circ \xor$, it suffices to show an upper bound on the spectral norm of the communication matrix of $f \circ \xor$.

Combining Theorem \ref{thm: forster} and Theorem \ref{lem: xor_eigenvalues}, we get
\begin{cor}\label{cor: xorforster}
Let $f: \bra{0, 1}^n \rightarrow \mathbb{R}$ be any real valued function and let $A$ denote the communication matrix of $f \circ \xor$.  Then,
\[
sr(A) \geq \frac{1}{\max\limits_{S \subseteq [n]}\abs{\wh{f}(S)}} \cdot \min_x|f(x)|
\]
\end{cor}
Thus, $sr(f \circ \xor) = 2^{\Omega(n)}$ for any $\bra{-1, 1}$ valued function with inverse exponential $l_\infty$ Fourier norm.

Note that we cannot use the outer function to be $\mod_p$ (for a constant $p$) in Corollary \ref{cor: xorforster}, since its principal Fourier coefficient is a constant (though sufficiently bounded away from 1, which we crucially require).
The following theorem allows us to ignore a subset of large Fourier coefficients, as long as their mass is not too large, which gives us a stronger condition for unbounded error hardness of \textsf{XOR} functions.

\begin{theorem}\label{thm: sufficient}
For any function $f: \bra{0, 1}^n \rightarrow \bra{-1, 1}$, and any collection of sets $\mathcal{S} \subseteq supp(\wh{f})$, if $\sum_{S \in \mathcal{S}}\abs{\wh{f}(S)} \leq 1 - \delta$, and $\max_{S \notin \mathcal{S}}\abs{\wh{f}(S)} \leq c$.
Then, $sr(f \circ \xor) \geq \frac{\delta}{c}$.
\end{theorem}

\begin{proof}
Define $f': \bra{0, 1}^n \rightarrow \mathbb{R}$ by $f'(x) = f(x) - \sum\limits_{S \in \mathcal{S}}\wh{f}(S)\chi_S(x)$.
Notice
\[
\min\limits_{x \in \bra{0, 1}^n}\abs{f'(x)} \geq 1 - \sum\limits_{S \in \mathcal{S}}\abs{\wh{f}(S)} \geq \delta
\]
Also note that $\forall S \in \mathcal{S}, \wh{f'}(S) = 0$, and $\forall S \notin \mathcal{S}, \wh{f'}(S) = \wh{f}(S)$.
Thus, $\max\limits_{S \subseteq [n]}\abs{\wh{f'}(S)} \leq c$.  It is easy to see that $f'$ sign agrees with $f$.
Thus, the sign rank of these functions agree by definition.
Using Corollary \ref{cor: xorforster}, we have
\begin{equation}\label{eqn: mainer}
sr(f \circ \xor) = sr(f' \circ \xor) \geq \frac{1}{\max\limits_{S \notin \mathcal{S}}\abs{\wh{f'}(S)}} \cdot \min\limits_{x}\abs{f'(x)} \geq \frac{\delta}{c}
\end{equation}
\end{proof}

Let us first recall the Complete Quadratic function, whose Fourier coefficients were analyzed by Bruck \cite{Bruck90}.
Define $\cq: \bra{0, 1}^n \rightarrow \bra{-1, 1}$ by
\[
\cq(x) = \mod_4^{\bra{0, 1}}(x)
\]
\begin{lemma}[\cite{Bruck90}]\label{lem: bruck}
For even $n$, $\abs{\wh{\cq}(S)} = 2^{-n/2}$ for all $S \subseteq [n]$.
For odd $n$, $\abs{\wh{\cq}(S)} \in \bra{0, 2^{-(n - 1)/2}}$ for all $S \subseteq [n]$.
\end{lemma}

\begin{theorem}\label{thm: odd_sign_rank}
For $m$ odd, and and $A \subseteq \bra{0, 1, \dots, m - 1}$ which is not the empty set or full set,
\[
U(\mod_m^A) = \Omega(n/m^2) - 2\log(m)
\]
\end{theorem}

\begin{proof}
In Theorem \ref{thm: sufficient}, use $\mathcal{S} = \emptyset$.
The values obtained using Claim \ref{claim: odd_fourier_coefficients} are $\delta = \frac{2}{m} - 2m \left(\cos\left(\frac{\pi}{2m}\right)\right)^{n}$, and $c = 2m \left(\cos\left(\frac{\pi}{2m}\right)\right)^{n}$.
Hence,
\begin{align*}
sr(\mod_m^A \circ \xor) & \geq \left(\frac{2}{m} - 2m \left(\cos\left(\frac{\pi}{2m}\right)\right)^{n}\right) \cdot \frac{1}{2m \left(\cos\left(\frac{\pi}{2m}\right)\right)^{n}}\\
\geq \frac{1}{m^2 \left(\cos\left(\frac{\pi}{2m}\right)\right)^{n}} - 1
\end{align*}
Using a standard series expansion for $\cos\theta$, and the fact that $1 - x \leq e^{-x}$ for all $x \in \mathbb{R}$, we get
\[
sr(\mod_m^A \circ \xor) \geq \frac{2^{\Omega(n/m^2)}}{m^2} - O(1)
\]
Thus, using the equivalence between sign rank and unbounded error communication complexity from Theorem \ref{thm: PS},
\[
U(\mod_m^A) = \Omega(n/m^2) - 2\log(m)
\]
\end{proof}

This already shows us that the unbounded error complexity of functions of the type $\mod_m^A$ are large when $m$ is odd, and $A$ is not the full set or empty set, for $m$ up to $O(n^{1/2 - \epsilon})$.
Note that one cannot use Theorem \ref{thm: sufficient} to prove a sign rank lower bound for $\mod_4^{\bra{0}}$, since $\abs{\wh{\mod_4^{\bra{0}}}\left(\emptyset\right)} + \abs{\wh{\mod_4^{\bra{0}}}\left([n]\right)} = 1$, which can be easily checked.  In Claim \ref{claim: mod4}, we also show hardness for the case when $m = 4$ and $A$ is not a simple accepting set.

In the analysis of our main claim (Theorem \ref{claim: main}), we will be concerned with the size of the input string.
For notational convenience, we add a subscript to $\mod_m^A$ which denotes the input size.
That is,
\[
\mod_{m, n}^A: \bra{0, 1}^n \rightarrow \bra{-1, 1}
\]
and we define it exactly the same as in Definition \ref{defn: mod}.

We denote the sumset $A + \bra{p} = \bra{a + p ~|~ a \in A}$ (the sums are modulo $m$, where $m$ is the period of the $\mod$ function we are interested in) by $A + p$ for convenience.

\begin{lemma}\label{lemma: shifting_parameters}
Suppose $\mod_{p, n}^{A'} = \mod_{m, n}^A \oplus \mod_{m, n}^{A + i}$ for some $p < m$, and any integer $i$.  Then,
\[
U(\mod_{m, n}^A) \geq \frac{U(\mod_{p, n-m}^{A'})}{2}
\]
\end{lemma}

We require the following simple, yet powerful lemma, the proof of which we omit.
\begin{lemma}[Folklore]\label{lemma: xorcomm}
For any functions $f, g: \bra{0, 1}^n \rightarrow \mathbb{R}$
\[
U(f \oplus g) \leq U(f) + U(g)
\]
\end{lemma}

\begin{proof}[Proof of Lemma \ref{lemma: shifting_parameters}]
Since $\mod_{p, n}^{A'} = \mod_{m, n}^A \oplus \mod_{m, n}^{A + i}$, applying Lemma \ref{lemma: xorcomm} gives us
\[
U(\mod_{p, n-m}^{A'}) \leq U(\mod_{m, n-m}^A) + U(\mod_{m, n-m}^{A + i})
\]
The first term on the right is at most $U(\mod_{m, n}^A)$ since we can just pad $m$ number of 0's each to Alice's and Bob's inputs and obtain a protocol (of the same cost) for $\mod_{m, n-m}^A$ given a protocol for $\mod_{m, n}^A$
The second term is also at most $U(\mod_{m, n}^A)$ for a similar reason. Pad $m - i$ number of 1's and $i$ number of 0's each to Alice's and Bob's inputs.
It is easy to see that $\mod_{m, n-m}^{A + i}(x, y) = -1$ if and only if $\mod_{m, n}^{A}(x', y') = -1$, where $x'$ and $y'$ are $x$ and $y$ padded with $m - i$ 1's and $i$ 0's respectively.
The lemma now follows.
\end{proof}

Let us analyze the unbounded error communication complexity of $\mod_4^A \circ \xor$ for various accepting sets $A$.  Note that if $A = \bra{0, 2}$ or $\bra{1, 3}$, then $\mod_4^A \circ \xor$ is
just parity or its negation respectively.  Its communication complexity is a constant in these cases.
Let us look at the other cases.
\begin{claim}\label{claim: mod4}
Suppose $A$ is not a simple accepting set.  Then, $U(\mod_4^A) = \Omega(n)$.
\end{claim}
\begin{proof}

\begin{enumerate}
   \item $A = \bra{0, 1}$.  Then, $\mod_4^A = \cq$, and by Lemma \ref{lem: bruck},
   \[
   U(\cq) \geq n/2
   \]
   \item $|A| = 2$, and $\mod_4^A$ does not represent parity (or its negation).  Then, this is clearly a translate of $\cq$, and
   \[
   U(\mod_{4, n}^A) \geq U(\mod_{4, n - 4}^{\bra{0, 1}}) \geq (n - 4)/2
   \]
   \item $A$ is non simple and does not fall in the previous 2 cases.  Without loss of generality, may assume $\abs{A} = 1$ because if it was 3, the complexity of $\mod_m^A$ is the same as $\mod_m^{A^c}$, and $\abs{A^c} = 1$.
   In this case, we can use Lemma \ref{lemma: shifting_parameters} to get
   \begin{align*}
   & U(\mod_4^A \oplus \mod_4^{A + 1}) \geq U(\mod_m^{A'})
   \end{align*}
   for some non simple $A'$ of size 2.  From the previous case, we conclude,
   \[
   U(\mod_{4, n}^A) \geq U(\mod_{4, n - 4}^{A'}) \geq \frac{\left((n - 4)/2\right) - 4}{2} = (n - 12)/4
   \]
   \end{enumerate}
\end{proof}

Recall our main theorem regarding unbounded error complexity (Theorem \ref{thm: actual_main}), which says that any function of the type $\mod_m^A \circ \xor$ for any non-simple $A$ is hard in the unbounded error communication model for values of $m$ up to $O(n^{1/2 - \epsilon})$.

\begin{theorem*}
For any integer $m \geq 3$, express $m = j2^k$ uniquely, where $j$ is either odd or 4, and $k$ is a positive integer.
Then for any non-simple $A$,
\[
U(\mod_{m, n}^A) \geq \Omega\left(\frac{n - km}{jm}\right) - \frac{2j\log j}{m}
\]
\end{theorem*}
Note that since $k$ is at most $\log(n)$, and $j$ is at most $m$,
this gives us an $n^{\Omega(1)}$ lower bound on the unbounded communication complexity of $\mod_m^A \circ \xor$ for any non-simple accepting set $A$, for $m$ as large as $O(n^{1/2 - \epsilon})$.

We require the following claim to prove Theorem \ref{thm: actual_main}.
\begin{claim}\label{claim: main}
For any integer $m \geq 3$, and for all representations $m = j2^k$ for some $j \geq 3$ and a positive integer $k$, and any non-simple $A \subseteq [m]$, we have
\[
U(\mod_{m, n}^A) \geq \frac{U(\mod_{j, n - km})}{2^k}
\]
\end{claim}
Let us first see how Claim \ref{claim: main} implies Theorem \ref{thm: actual_main}.  Recall that Theorem \ref{thm: odd_sign_rank} gave us
\[
U(\mod_{j, n}) = \Omega(n/j^2) - 2\log(j)
\]
This, along with Claim \ref{claim: mod4} and Claim \ref{claim: main}, implies that if $m = j2^k$ where $j$ is either 4 or odd,
\begin{align*}
U(\mod_{m, n}^A) \geq \frac{U(\mod_{j, n - km})}{2^k} \geq \frac{\Omega\left(\frac{(n - km)}{j^2}\right) - 2\log(j)}{m/j} \geq \Omega\left(\frac{n - km}{jm}\right) - \frac{2j\log j}{m}
\end{align*}
Let us now prove Claim \ref{claim: main}.
\begin{proof}
We prove this by induction on $m$.

\begin{enumerate}
 \item The base cases are when $m$ is odd.
 In this case, the hypothesis is trivially true since $m = j2^k$ can only imply $j = m, k = 0$.
 \item Suppose $m = 2p$, where $p$ is odd.  Let $a = xy$ denote the characteristic vector of the accepting set $A$,
 where $x$ corresponds to the first $p$ elements, and $y$ the last $p$ elements.  We interchangeably use the notation $\mod_m^A$ and $\mod_m^a$ when $a$ is the characteristic vector of the set $A$.  Our assumption is that
 $a$ is not the all 0, or all 1, or the parity (negation of parity) vector.  Let $x \oplus y$ denote the bitwise $\xor$ of $x$ and $y$.
 \begin{enumerate}
  \item   Suppose $x \oplus y$ is neither the all 0 or all 1 vector.  Since $x \oplus y$ does not represent a simple accepting set $A$, in this case, $\mod_m^A \oplus \mod_m^{A + p} = \mod_p^{x \oplus y}$.  By Lemma \ref{lemma: shifting_parameters},
  \[
  U(\mod_{m, n}^A) \geq \frac{U(\mod_{p, n - m}^{x \oplus y})}{2}
  \]
  \item If $x \oplus y$ is the all 0 vector, then $x = y$, and neither of them are all 0 or all 1.
  This means $\mod_m^a = \mod_p^{x}$.
  \item If $x \oplus y$ is the all 1 vector, this means $y = x^c$.
  Consider $A' = A + {1}$.  One may verify that $A \oplus A'$ has characteristic vector $a''= bb$.
  \begin{enumerate}
   \item If $b$ is not the all 0 or all 1 string, $\mod_p^b = \mod_m^A \oplus \mod_m^{A + 1}$.  Use Lemma \ref{lemma: shifting_parameters} and conclude
   \[
   U(\mod_{m, n}^A) \geq \frac{U(\mod_{p, n - m}^{b})}{2}
   \]
    \item It is easy to check that $b$ can never be the all 0 vector.
    \item Close inspection reveals that if $b$ is the all 1 vector, then the original vector $a$
    must represent parity or its negation, which was not the case by assumption.
   \end{enumerate}
 \end{enumerate}
 
 \item Suppose $m = 2k$, where $k$ is even.  Again, let $a = xy$, where $a$ is the characteristic vector of accepting set $A$.
 \begin{enumerate}
  \item If $x \oplus y$ is neither the all 0 string, all 1 string, nor does it represent parity (or its negation), then $\mod_m^A \oplus \mod_m^{A + k} = \mod_k^{x \oplus y}$.  By Lemma \ref{lemma: shifting_parameters},
  \[
  U(\mod_{m, n}^A) \geq \frac{U(\mod_{k, n - m}^{x \oplus y})}{2}
  \]
  By the induction hypothesis, the claim is true for $\mod_{k, n - m}^{x \oplus y}$.  It is easy to see that this implies the claim for $\mod_{m, n}^A$.
  \item If $x \oplus y$ is the all 0 vector, then $x = y$, and neither of them are all 0 or all 1.
  This means $\mod_m^a$ is the same as $\mod_k^{x}$.
  \item\label{item: complement} If $x \oplus y$ is the all 1 vector, this means $y = x^c$.
  Consider $A' = A + 1$.  One may verify that $A \oplus A'$ has a characteristic vector of the form $a''= bb$.
   \begin{enumerate}
    \item If $b$ is neither the all 0 or all 1 string, nor does it represent parity (or its negation), then $\mod_k^b = \mod_m^A \oplus \mod_m^{A + 1}$.  Use Lemma \ref{lemma: shifting_parameters} and conclude
    \[
    U(\mod_{m, n}^A) \geq \frac{U(\mod_{k, n - m}^{b})}{2}
    \]
    By the induction hypothesis, the theorem is true for $\mod_k^b$ since $b$ does not represent the all 0, all 1, or parity (or its negation) string.
    The theorem now follows easily for $\mod_{m, n}^A$.
    \item It is easy to check that $b$ can never be the all 0  or all 1 vector.
    \item One may check that $b$ can be the parity (or its negation) vector only if $k \equiv 2 (\text{mod}~4)$, and $A$
    must have represented $\cq$ (or a translate of it by at most 2) which we know to be hard.
    In this case, we obtain
    \[
    U(\mod_m^a) = \Omega(n)
    \]
   \end{enumerate}
  \item If $x \oplus y$ represents the parity (or negation of parity) vector, then consider $A' = A + 2$.  It is simple to verify that the characteristic vector of $A \oplus A'$ is of the form $zz$.
   \begin{enumerate}
    \item If $z$ is not the all 0 or all 1 string, or does not represent parity (or its negation), then we have $\mod_m^A \oplus \mod_m^{A + 2} = \mod_k^{z}$.  Use Lemma \ref{lemma: shifting_parameters} to say
    \[
    U(\mod_{m, n}^A) \geq \frac{U(\mod_{k, n - m}^{z})}{2}
    \]
    The claim now follows because of the induction hypothesis.
    \item One may verify (by considering cases when $k$ has residue either 0 or 2 modulo 4) that $z$ cannot be the all 0 or all 1 string.
    \item If $z$ represents parity or its negation, then it can be checked that the only case when this occurs is when $A$ represented a non simple accepting set, say $X$, modulo 4.
    Thus,
    \[
    U(\mod_m^a) = U(\mod_4^X) = \Omega(n)
    \]
   \end{enumerate}
  \end{enumerate}
\end{enumerate}
\end{proof}

\subsection{An upper bound}\label{sec: upper_bound}
In this section, we show that for any symmetric function $f: \bra{-1, 1}^n \rightarrow \bra{-1, 1}$, the $\PP$ complexity of $f \circ \xor$ is upper bounded by essentially $\deg_{oe}(f)$.
Our proof follows along the lines of Zhang \cite{Zhang91} who shows that a symmetric function with small odd-even degree has a small Threshold of Parity circuit representation.

\begin{theorem}\label{thm: upper_bound}
Suppose $f: \bra{-1, 1}^n \rightarrow \bra{-1, 1}$ is a symmetric function defined by the predicate $D_f: \bra{0, 1, \dots, n} \rightarrow \bra{-1, 1}$.  Say the odd-even degree of $f$ equals $k$ and $n$ is even.
Then,
\[
\PP(f \circ \xor) = O(k\log n)
\]
\end{theorem}
\begin{proof}

Define $S_{even} = \bra{i \in \bra{0, 2, \dots, n} : D_f(i) \neq D_f(i + 2)}$, and define $S_{odd} = \bra{i \in \bra{1, 3, \dots, n - 1} : D_f(i) \neq D_f(i + 2)}$.
By our assumption, $\abs{S_{even}}, \abs{S_{odd}} \leq k$.

Consider the polynomials $p_{even}, p_{odd} : \bra{-1, 1}^n \rightarrow \R$ defined by
\[
p_{even}(x) = D_f(0) \cdot \prod_{i \in S_{even}} \left(n - 2i + 1 - \left(\sum_{j = 1}^n x_j\right)\right)
\]
and
\[
p_{odd}(x) = D_f(1) \cdot \prod_{i \in S_{odd}} \left(n - 2i + 1 - \left(\sum_{j = 1}^n x_j\right)\right)
\]
The polynomial $p : \bra{-1, 1}^n \rightarrow \R$ defined by
\[
p(x) = (1 + \chi_{[n]}(x))p_{even}(x) + (1 - \chi_{[n]}(x))p_{odd}(x)
\]
sign represents $f$ on $\bra{-1, 1}^n$.

We now use the simple observations that $wt(q_1 \cdot q_2) \leq wt(q_1) \cdot wt(q_2)$ and $wt(q_1 + q_2) \leq wt(q_1) + wt(q_2)$.
Thus,
\begin{align*}
wt(p) & \leq 2wt(p_{even}) + 2wt(p_{odd})\\
& \leq 2(2n)^k + 2(2n)^k\\
& \leq 4 (2n)^k
\end{align*}

Note that all the coefficients of $p$ are integer valued.  Thus, the polynomial $p' = \frac{p}{wt(p)}$ is a polynomial of weight 1, which sign represents $f$ with margin at least $\frac{1}{wt(p)}$.
By Theorem \ref{thm: equiv} and Theorem \ref{thm: klauck}, 
\[
\PP(f \circ \xor) \leq O(\log(wt(p))) \leq O(k\log n)
\]

\end{proof}

\subsection{Circuits}\label{sec: ckts}

In this section, we show how to obtain a size lower bound on a restricted class of threshold circuits computing $\mod_m^A \circ \xor$ for non simple $A$.
Forster et al.~\cite{FKLMSS01} noted that sign rank lower bounds also yield lower bounds against $\THR \circ \MAJ$ circuits.
In fact, it yields lower bounds for the class $\THR \circ \mathsf{L_{comm}}$ where $\mathsf{L_{comm}}$ denotes any gate with low deterministic communication complexity.
We show the following.

\begin{theorem}\label{thm: ckt_lb}
Any $\THR \circ C$ circuit computing $\mod_m^A \circ \xor$ must have size
\[
s \geq 2^{\Omega\left(\frac{n - km}{jm}\right) - \frac{2j\log j}{m} - c}
\]
where $c$ is the deterministic communication complexity of $C$, and $m = j2^k$ is the unique representation of $m \geq 3$, where $j$ is either odd or 4, and $k$ is a positive integer.
\end{theorem}

\begin{proof}
The rank of the communication matrix of each $C$ gate is at most $c$, thus the sign rank of a function computed by a $\THR \circ C$ circuit is at most $sc$, where $s$ is
the size of the circuit.
Theorem \ref{thm: actual_main} and Theorem \ref{thm: PS} tells us that $sr(\mod_m^A \circ \xor) \geq 2^{\Omega\left(\frac{n - km}{jm}\right) - \frac{2j\log j}{m}}$, where $m = j2^k$, and $j$ is either 4 or odd.
Thus,
\[
sc \geq 2^{\Omega\left(\frac{n - km}{jm}\right) - \frac{2j\log j}{m}} \implies s \geq 2^{\Omega\left(\frac{n - km}{jm}\right) - \frac{2j\log j}{m} - c}
\]
\end{proof}
Thus, we obtain that for $m$ up to $O(n^{1/2 - \epsilon})$, and any non-simple $A$, $\mod_m^A \circ \xor$ is not in subexponential sized $\THR \circ \MAJ$.
A similar argument shows that $\mod_m^A \circ \xor$ is not even in subexponential size $\THR \circ \mathsf{SYM}$, where $\mathsf{SYM}$ denotes the class of all symmetric functions.
This is because all symmetric functions have deterministic communication complexity bounded above by $O(\log(n))$.

This generalizes one particular result of Krause and Pudlak \cite{KP97}, and of Zhang \cite{Zhang91} which state that $\mod_m^{\bra{0}} \notin \THR \circ \mathsf{PAR}$, where $\mathsf{PAR}$ denotes the class of all parity gates.
This is because we have shown that $\mod_m^A \circ \xor \notin \THR \circ \mathsf{SYM}$,
which implies $\mod_m^A \circ \xor \notin \THR \circ \mathsf{PAR}$.  This implies $\mod_m^A \notin \THR \circ \mathsf{PAR}$.

\section{Acknowledgements}

We thank anonymous reviewers for providing invaluable comments regarding the presentation of parts of this paper.
We thank Justin Thaler for pointers regarding the connection between margin and threshold weight, and bringing the recent paper of Hatami and Qian \cite{HQ17} to our notice.

\bibliography{bibo}

\begin{thebibliography}{10}

\bibitem{ACFN15}
Anil Ada, Arkadev Chattopadhyay, Omar Fawzi, and Phuong Nguyen.
\newblock The {NOF} multiparty communication complexity of composed functions.
\newblock {\em Computational Complexity}, 24(3):645--694, 2015.

\bibitem{AFH12}
Anil Ada, Omar Fawzi, and Hamed Hatami.
\newblock Spectral norm of symmetric functions.
\newblock In {\em Approximation, Randomization, and Combinatorial Optimization.
  Algorithms and Techniques - 15th International Workshop, {APPROX} 2012, and
  16th International Workshop, {RANDOM} 2012, Cambridge, MA, USA, August 15-17,
  2012. Proceedings}, pages 338--349, 2012.

\bibitem{BFS86}
L{\'{a}}szl{\'{o}} Babai, Peter Frankl, and Janos Simon.
\newblock Complexity classes in communication complexity theory (preliminary
  version).
\newblock In {\em 27th Annual Symposium on Foundations of Computer Science,
  Toronto, Canada, 27-29 October 1986}, pages 337--347, 1986.

\bibitem{BBG14}
Eric Blais, Joshua Brody, and Badih Ghazi.
\newblock The information complexity of hamming distance.
\newblock In {\em Approximation, Randomization, and Combinatorial Optimization.
  Algorithms and Techniques, {APPROX/RANDOM} 2014, September 4-6, 2014,
  Barcelona, Spain}, pages 465--489, 2014.

\bibitem{Bourgain05}
Jean Bourgain.
\newblock Estimation of certain exponential sums arising in complexity theory.
\newblock {\em Comptes Rendus Mathematique}, 340(9):627--631, 2005.

\bibitem{Bruck90}
Jehoshua Bruck.
\newblock Harmonic analysis of polynomial threshold functions.
\newblock {\em {SIAM} J. Discrete Math.}, 3(2):168--177, 1990.

\bibitem{BVW07}
Harry Buhrman, Nikolay Vereshchagin, and Ronald de~Wolf.
\newblock On computation and communication with small bias.
\newblock In {\em Proceedings of the Twenty-Second Annual IEEE Conference on
  Computational Complexity}, CCC '07, pages 24--32. IEEE Computer Society,
  2007.

\bibitem{BT15}
Mark Bun and Justin Thaler.
\newblock Hardness amplification and the approximate degree of constant-depth
  circuits.
\newblock In {\em Automata, Languages, and Programming - 42nd International
  Colloquium, {ICALP} 2015, Kyoto, Japan, July 6-10, 2015, Proceedings, Part
  {I}}, pages 268--280, 2015.

\bibitem{BT17}
Mark Bun and Justin Thaler.
\newblock A nearly optimal lower bound on the approximate degree of
  ac\({}^{\mbox{0}}\).
\newblock {\em CoRR}, abs/1703.05784, 2017.

\bibitem{Cha07}
Arkadev Chattopadhyay.
\newblock Discrepancy and the power of bottom fan-in in depth-three circuits.
\newblock In {\em 48th Annual {IEEE} Symposium on Foundations of Computer
  Science {(FOCS} 2007), October 20-23, 2007, Providence, RI, USA,
  Proceedings}, pages 449--458, 2007.

\bibitem{Cha09}
Arkadev Chattopadhyay.
\newblock {\em Circuits, Communication and Polynomials}.
\newblock PhD thesis, McGill University, 2009.

\bibitem{CA08}
Arkadev Chattopadhyay and Anil Ada.
\newblock Multiparty communication complexity of disjointness.
\newblock {\em Electronic Colloquium on Computational Complexity {(ECCC)}},
  15(002), 2008.

\bibitem{CGPT06}
Arkadev Chattopadhyay, Navin Goyal, Pavel Pudl{\'{a}}k, and Denis
  Th{\'{e}}rien.
\newblock Lower bounds for circuits with mod{\_}m gates.
\newblock In {\em 47th Annual {IEEE} Symposium on Foundations of Computer
  Science {(FOCS} 2006), 21-24 October 2006, Berkeley, California, USA,
  Proceedings}, pages 709--718, 2006.

\bibitem{Forster01}
J{\"{u}}rgen Forster.
\newblock A linear lower bound on the unbounded error probabilistic
  communication complexity.
\newblock In {\em Proceedings of the 16th Annual {IEEE} Conference on
  Computational Complexity, Chicago, Illinois, USA, June 18-21, 2001}, pages
  100--106, 2001.

\bibitem{FKLMSS01}
J{\"{u}}rgen Forster, Matthias Krause, Satyanarayana~V. Lokam, Rustam
  Mubarakzjanov, Niels Schmitt, and Hans~Ulrich Simon.
\newblock Relations between communication complexity, linear arrangements, and
  computational complexity.
\newblock In {\em {FST} {TCS} 2001: Foundations of Software Technology and
  Theoretical Computer Science, 21st Conference, Bangalore, India, December
  13-15, 2001, Proceedings}, pages 171--182, 2001.

\bibitem{GHR92}
Mikael Goldmann, Johan H{\aa}stad, and Alexander~A. Razborov.
\newblock Majority gates {VS.} general weighted threshold gates.
\newblock {\em Computational Complexity}, 2:277--300, 1992.

\bibitem{HHL16}
H.~Hatami, K.~Hosseini, and S.~Lovett.
\newblock Structure of protocols for {XOR} functions.
\newblock {\em Electronic Colloquium on Computational Complexity {(ECCC)}},
  23:44, 2016.

\bibitem{HQ17}
Hamed Hatami and Yingjie Qian.
\newblock The unbounded-error communication complexity of symmetric xor
  functions.
\newblock {\em Arxiv}, 2017.

\bibitem{KS92}
Bala Kalyanasundaram and Georg Schnitger.
\newblock The probabilistic communication complexity of set intersection.
\newblock {\em {SIAM} J. Discrete Math.}, 5(4):545--557, 1992.

\bibitem{Klauck07}
Hartmut Klauck.
\newblock Lower bounds for quantum communication complexity.
\newblock {\em {SIAM} J. Comput.}, 37(1):20--46, 2007.

\bibitem{Krause06}
Matthias Krause.
\newblock On the computational power of boolean decision lists.
\newblock {\em Computational Complexity}, 14(4):362--375, 2006.

\bibitem{KP97}
Matthias Krause and Pavel Pudl{\'{a}}k.
\newblock On the computational power of depth-2 circuits with threshold and
  modulo gates.
\newblock {\em Theor. Comput. Sci.}, 174(1-2):137--156, 1997.

\bibitem{KN97}
Eyal Kushilevitz and Noam Nisan.
\newblock {\em Communication complexity}.
\newblock Cambridge University Press, 1997.

\bibitem{LS09a}
Troy Lee and Adi Shraibman.
\newblock Disjointness is hard in the multiparty number-on-the-forehead model.
\newblock {\em Computational Complexity}, 18(2):309--336, 2009.

\bibitem{LS09}
Troy Lee and Adi Shraibman.
\newblock Lower bounds in communication complexity.
\newblock {\em Foundations and Trends in Theoretical Computer Science},
  3(4):263--398, 2009.

\bibitem{MP69}
Marvin Minsky and Seymour Papert.
\newblock {\em Perceptrons - an introduction to computational geometry}.
\newblock {MIT} Press, 1987.

\bibitem{Paturi92}
Ramamohan Paturi.
\newblock On the degree of polynomials that approximate symmetric boolean
  functions (preliminary version).
\newblock In {\em Proceedings of the 24th Annual {ACM} Symposium on Theory of
  Computing, May 4-6, 1992, Victoria, British Columbia, Canada}, pages
  468--474, 1992.

\bibitem{PS86}
Ramamohan Paturi and Janos Simon.
\newblock Probabilistic communication complexity.
\newblock {\em J. Comput. Syst. Sci.}, 33(1):106--123, 1986.

\bibitem{Razborov92}
Alexander~A. Razborov.
\newblock On the distributional complexity of disjointness.
\newblock {\em Theor. Comput. Sci.}, 106(2):385--390, 1992.

\bibitem{RS10}
Alexander~A. Razborov and Alexander~A. Sherstov.
\newblock The sign-rank of {AC}\({}^{\mbox{0}}\).
\newblock {\em {SIAM} J. Comput.}, 39(5):1833--1855, 2010.

\bibitem{She08}
Alexander~A. Sherstov.
\newblock Halfspace matrices.
\newblock {\em Computational Complexity}, 17(2):149--178, 2008.

\bibitem{She09}
Alexander~A. Sherstov.
\newblock Separating {AC}\({}^{\mbox{0}}\) from depth-2 majority circuits.
\newblock {\em {SIAM} J. Comput.}, 38(6):2113--2129, 2009.

\bibitem{She11}
Alexander~A. Sherstov.
\newblock The pattern matrix method.
\newblock {\em {SIAM} J. Comput.}, 40(6):1969--2000, 2011.

\bibitem{Sherstov11a}
Alexander~A. Sherstov.
\newblock The unbounded-error communication complexity of symmetric functions.
\newblock {\em Combinatorica}, 31(5):583--614, 2011.

\bibitem{She14}
Alexander~A. Sherstov.
\newblock Communication lower bounds using directional derivatives.
\newblock {\em J. ACM}, 61(6):34:1--34:71, 2014.

\bibitem{She16}
Alexander~A. Sherstov.
\newblock On multiparty communication with large versus unbounded error.
\newblock {\em Electronic Colloquium on Computational Complexity {(ECCC)}},
  2016.

\bibitem{ShiZ09}
Yaoyun Shi and Yufan Zhu.
\newblock Quantum communication complexity of block-composed functions.
\newblock {\em Quantum Information {\&} Computation}, 9(5):444--460, 2009.

\bibitem{Thaler16}
Justin Thaler.
\newblock Lower bounds for the approximate degree of block-composed functions.
\newblock In {\em 43rd International Colloquium on Automata, Languages, and
  Programming, {ICALP} 2016, July 11-15, 2016, Rome, Italy}, pages 17:1--17:15,
  2016.

\bibitem{Yao79}
Andrew~Chi{-}Chih Yao.
\newblock Some complexity questions related to distributive computing
  (preliminary report).
\newblock In {\em Proceedings of the 11h Annual {ACM} Symposium on Theory of
  Computing, April 30 - May 2, 1979, Atlanta, Georgia, {USA}}, pages 209--213,
  1979.

\bibitem{Zhang91}
Zhi-Li Zhang.
\newblock Complexity of symmetric functions in perceptron-like models.
\newblock Master's thesis, University of Massachusetts at Amherst, 1992.

\bibitem{SZ09}
Zhiqiang Zhang and Yaoyun Shi.
\newblock Communication complexities of symmetric {XOR} functions.
\newblock {\em Quantum Information {\&} Computation}, 9(3):255--263, 2009.

\end{thebibliography}

\end{document}